\documentclass[a4paper,UKenglish,cleveref,autoref,thm-restate]{article}
\usepackage{graphicx} % Required for inserting images
\usepackage[margin=1in]{geometry}
\usepackage{amsthm}
\usepackage{amssymb,amsmath}  
\usepackage{mathtools}
\usepackage[colorlinks=true, allcolors=blue]{hyperref}
\usepackage[capitalize]{cleveref}
\usepackage{todonotes}
\usepackage{placeins}
\usepackage{subcaption}
\usepackage{booktabs}
\usepackage{authblk}

\usepackage{enumitem}
\setlist[1]{labelindent=\parindent}
\setlist[enumerate]{label=(\arabic*)}

\usepackage{xcolor}

\title{When is String Reconstruction using de Bruijn Graphs Hard?} %TODO Please add

\author[1,2]{Ben Bals}
\author[1,3]{Sebastiaan van Krieken}
\author[1,2]{Solon P. Pissis}
\author[1,2]{Leen Stougie}
\author[1]{\\Hilde Verbeek}
\affil[1]{CWI, Amsterdam, The Netherlands}
\affil[2]{Vrije Universiteit, Amsterdam, The Netherlands}
\affil[3]{TU Delft, The Netherlands}

\newtheorem{theorem}{Theorem}[section]
\newtheorem{corollary}[theorem]{Corollary}
\newtheorem{lemma}[theorem]{Lemma}
\newtheorem{definition}[theorem]{Definition}

\newtheorem{remark}[theorem]{Remark}

\usepackage[bibliography=common,appendix=inline]{apxproof}
\usepackage{todonotes}
\usepackage{amsthm}
\usepackage{thmtools}
\usepackage{placeins}
\usepackage[normalem]{ulem}
\usepackage{bm}
\usepackage{tabularx}
\usepackage{makecell}
\usepackage{xspace}
\usepackage{appendix}

\crefname{enumi}{condition}{conditions}

\ifdefined\N\else\DeclareMathOperator{\N}{\mathbb{N}}\fi
\ifdefined\Z\else\DeclareMathOperator{\Z}{\mathbb{Z}}\fi
\ifdefined\R\else\DeclareMathOperator{\R}{\mathbb{R}}\fi

\newcommand{\BigO}{\mathcal{O}}

\newcommand{\setsize}[1]{\left|#1\right|}

\newif\ifpaper
\definecolor[named]{benGray}{rgb}{0.61,0.61,0.63}
\definecolor[named]{maybeDone}{cmyk}{0.45,0,0.50,0.32}

\newcommand{\dirProblem}{\mathrm{diET}\xspace}
\newcommand{\dirProblemCost}{\mathrm{dicET}\xspace}

\newcommand{\undirProblemCost}{\mathrm{uicET}\xspace}

\newcommand{\iWidth}{w}
\newcommand{\fptProblem}{\dirProblem[w]\xspace}

\newcommand{\fptProblemCost}{\dirProblemCost[w]\xspace}
\newcommand{\fpt}{\mathsf{FPT}}
\newcommand{\PClass}{\mathsf{P}}
\newcommand{\NP}{\mathsf{NP}}

\newcommand{\SharpP}{\#\PClass}
\newcommand{\cmax}{c_\mathrm{budget}}

\newcommand{\tail}[1]{\mathrm{tail}(e)}
\newcommand{\letter}[1]{\texttt{#1}}
\newcommand{\letA}{\letter{A}}

\newcommand{\letT}{\letter{T}}
\newcommand{\id}[1]{\mathrm{id}_{#1}}

\newcommand{\fptBase}{\lambda}
\newcommand{\intervals}[1]{\mathcal{I}_{#1}}
\DeclareMathOperator{\imw}{imw}
\newcommand{\str}{\alpha}
\def\dd{\mathinner{.\,.}}

\usepackage{caption}
\usepackage{subcaption}
\usepackage{mathtools}

\usepackage[ruled]{algorithm2e}

\usepackage[skins]{tcolorbox}
\usepackage{booktabs}

\definecolor{CWIBlue}{rgb}{0.38, 0.478, 0.56}
\definecolor{CWIDGreen}{rgb}{0.372, 0.494, 0.419}
\definecolor{CWILGreen}{rgb}{0.603, 0.717, 0.486}
\definecolor{CWIRed}{rgb}{0.686, 0.188, 0.262}

\usepackage{etoolbox}

\makeatletter
\def\@ifundefinedcolor#1{\@ifundefined{\string\color@#1}}

\newcommand{\accentColor}{%
  \@ifundefinedcolor{lipicsYellow}{CWIRed}{lipicsYellow}}
\makeatother

\makeatletter

\newcolumntype{\expand}{}
\long\@namedef{NC@rewrite@\string\expand}{\expandafter\NC@find}

\NewEnviron{problem}[2][]{%
  \def\problem@arg{#1}%
  \def\problem@framed{framed}%
  \def\problem@hline{\hline}%
\def\problem@tablelayout{|>{\bfseries}lX|c}%
\def\problem@title{\multicolumn{2}{|%
  >{\raisebox{-\fboxsep}}%
  p{\dimexpr\textwidth-4\fboxsep-2\arrayrulewidth\relax}%
  |}{%
      \textcolor{\accentColor}{\bfseries$\blacktriangleright$} {\bfseries Problem.} \textsc{#2}%
  }}%
  \bigskip\par\noindent%
  \begin{tabularx}{\textwidth}{\expand\problem@tablelayout}%
    \problem@hline%
    \problem@title\\[1.5\fboxsep]%
    \BODY%\hfill\textcolor{stroke1}{\bfseries$\blacktriangleleft$}
    \\\problem@hline%
  \end{tabularx}%
  \medskip\par%
  \vspace{.5mm}
}
\makeatother

\papertrue

\begin{document}

\maketitle

\begin{abstract}
The reduction of the fragment assembly problem to (variations of) the classical Eulerian trail problem [Pevzner et al., PNAS~2001] has led to remarkable progress in genome assembly. This reduction employs the notion of \emph{de Bruijn graph} $G=(V,E)$ of order $k$ over an alphabet $\Sigma$. A single Eulerian trail in $G$ represents a \emph{candidate} genome reconstruction. Bernardini et al.~have also introduced the complementary idea in data privacy [ALENEX 2020] based on $z$-anonymity. Let $S$ be a private string that we would like to release, preventing, however, its full reconstruction. For a privacy threshold $z>0$, we compute the largest $k$ for which there exist at least $z$ Eulerian trails in the order-$k$ de Bruijn graph of $S$, and release a string $S'$ obtained via a random Eulerian trail.

The pressing question is: How hard is it to reconstruct a best string from a de Bruijn graph given a function that models domain knowledge? Such a function maps every length-$k$ string to an \emph{interval of positions} where it may occur in the reconstructed string. By the above reduction to de Bruijn graphs, the latter function translates into a function $c$ mapping every edge to an interval  where it may occur in an Eulerian trail. This gives rise to the following basic problem on graphs:
\begin{center}
\emph{Given an instance $(G,c)$, can we efficiently compute an Eulerian trail respecting $c$?} 
\end{center}
Hannenhalli et al.~[CABIOS 1996] formalized this problem and showed that it is \(\NP\)-complete. Ben{-}Dor et al.~[J.~Comput.~Biol.~2002] showed that it is \(\NP\)-complete, even on de Bruin graphs with $|\Sigma|=4$. In this work, we settle the lower-bound side of this problem by showing that finding a $c$-respecting Eulerian trail in de Bruijn graphs over alphabets of size $2$ is $\NP$-complete.

We then shift our focus to \emph{parametrization} aiming to capture the quality of our domain knowledge in the complexity.
Ben{-}Dor et al.~developed an algorithm to solve the problem on de Bruijn graphs in $\BigO(m \cdot \iWidth^{1.5} 4^{\iWidth})$ time, where $m=|E|$ and $w$ is the \emph{maximum interval length} over all edges in $E$. Bumpus and Meeks [Algorithmica 2023] later rediscovered the same algorithm on temporal graphs, which highlights the relevance of this problem in other contexts. Our central contribution is showing how combinatorial insights lead to \emph{exponential-time} improvements over the state-of-the-art algorithm. In particular, for the important class of de Bruijn graphs, we develop an algorithm parametrized by $w (\log w+1) /(k-1)$: for a de Bruijn graph of order $k$, it runs in \(\BigO(mw \cdot 2^{\frac{w (\log w+1)}{k-1}})\) time. Our result improves on the state of the art by roughly an exponent of $(\log w +1)/(k-1)$. The existing algorithms have a natural interpretation for string reconstruction: when for each length-$k$ string, we know a small range of positions it must lie in, string reconstruction can be solved in linear time. Our improved algorithm shows that \emph{it is enough when the range of positions is small relative to $k$}. 

We then generalize both the existing and our novel $\fpt$ algorithm by allowing the cost at every position of an interval to vary. In this optimization version, our hardness result translates into inapproximability and the $\fpt$ algorithms work with a slight extension. Surprisingly, even in this more general setting, we extend the $\fpt$ algorithms to count and enumerate the min-cost Eulerian trails. The counting result has direct applications in the data privacy framework of Bernardini et al.
\end{abstract}

\newpage

\section{Introduction\label{sec:intro}}
One of the most important algorithmic tasks in bioinformatics
is that of \emph{genome assembly} (or \emph{fragment assembly})~\cite{DBLP:journals/jcb/TomescuM17,DBLP:conf/cpm/CairoMART17,DBLP:journals/bib/Medvedev19,DBLP:journals/talg/CairoRTZ24}: the process of taking a large number of short DNA fragments and putting them back together to create a representation of the original chromosomes from which the DNA originated. The textbook reduction of the fragment assembly problem to (variations of) the classical Eulerian trail problem~\cite{Assembly2001} has led to remarkable progress in the past three decades. This reduction is based on the notion of \emph{de Bruijn graph} (dBG, in short).

Let $S=S[1]\ldots S[|S|]=S[1\dd |S|]$ be a \emph{string} of length $|S|$ over an \emph{alphabet} $\Sigma$.
We fix a collection $\mathcal{S}$ of strings over $\Sigma$
and define the \emph{order-$k$ de Bruijn graph} of $\mathcal{S}$ as a directed graph, denoted by $G_{\mathcal{S},k}=(V,E)$, where $V$ is the set of length-$(k-1)$ substrings of the strings in $\mathcal{S}$ and $E$ has an edge $(u,v)$ if and only if $u[1]\cdot v=u\cdot v[k-1]$ and $u[1]\cdot v$ occurs in some string $S$ of $\mathcal{S}$.
In applications, we often consider a de Bruijn \emph{multigraph} where the multiplicity of an edge is exactly the total number of these occurrences in the strings in \(\mathcal S\).
Then, a single Eulerian trail in $G_{\mathcal{S},k}$ represents a \emph{candidate} string reconstruction~\cite{Assembly2001};  see \Cref{fig:dBG}.

\begin{figure}[ht]
    \centering
    \includegraphics[width=0.7\textwidth]{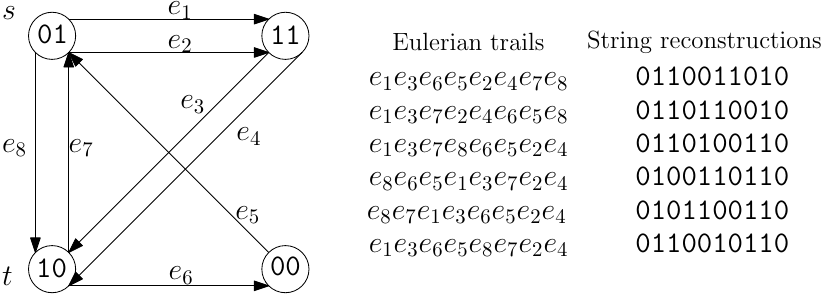}
    \caption{The de Bruijn multigraph $G_{\mathcal{S},k}=(V,E)$ (left), the set of node-distinct Eulerian trails from $s$ to $t$ (middle), and the corresponding set of string reconstructions (right) for
    the string collection $\mathcal{S}= \texttt{001}, \texttt{010}, \texttt{011}, \texttt{011}, \texttt{100}, \texttt{101}, \texttt{110}, \texttt{110}$, over the alphabet $\Sigma=\{\texttt{0},\texttt{1}\}$, and $k=3$.}
    \label{fig:dBG}
\end{figure}

Bernardini et al.~have introduced the complementary idea in data privacy~\cite{DBLP:conf/alenex/0001CFLP20} based on the \emph{$z$-anonymity} privacy property~\cite{DBLP:conf/pods/SamaratiS98,DBLP:journals/ijufks/Sweene02}. Let $S$ be a \emph{private} string that we would like to release for data analysis, preventing, however, its full reconstruction. For a privacy threshold $z>0$, we compute the largest $k$ for which there exist at least $z$ node-distinct Eulerian trails in the order-$k$ dBG of $\mathcal{S}=\{S\}$, and release a string $S'$ obtained via a random Eulerian trail. In \Cref{fig:dBG}, we have $6$ node-distinct Eulerian trails (and thus $6$ distinct strings) corresponding to $\mathcal{S}$. Under the $z$-anonymity assumption, one cannot know which of the $6$ strings is $S$ unless they can rely on some additional information about $S$, such as on \emph{domain knowledge}.

The pressing question arising from these applications is thus: How hard is it to reconstruct a best string from $G_{\mathcal{S},k}$ given a function modeling domain knowledge? 
This function maps every length-$k$ string to an \emph{interval of positions} where it may occur in the reconstructed string. By the above reduction, this translates into a function $c$ mapping every edge to an interval where it may occur in an Eulerian trail, raising the following basic graph problem:

\begin{center}
\emph{Given an instance $(G,c)$, can we efficiently compute an Eulerian trail respecting $c$?}    
\end{center}

Although Hannenhalli et al.~\cite{DBLP:journals/bioinformatics/HannenhalliFLSP96} formalized this basic problem in the context of dBGs, it has applications in temporal graphs~\cite{DBLP:journals/algorithmica/BumpusM23,DBLP:journals/algorithmica/MarinoS23,DBLP:journals/tcs/MichailS16} and other networks~\cite{DBLP:journals/networks/DrorST87,
%DBLP:journals/orl/GhianiI00,DBLP:journals/eor/KortewegV06, 
WANG2002375, DBLP:conf/mfcs/KupfermanV16}. We denote it here by \(\dirProblem\) (see \Cref{sec:setting} for a formal definition). Hannenhalli et al.~showed that \(\dirProblem\) is $\NP$-complete when each node of $G$ has in- and out-degree at most two. Even if their motivating application was fragment assembly and thus based on dBGs, their result was shown on general directed graphs. 
Ben{-}Dor et al.~\cite{DBLP:journals/jcb/Ben-DorPSS02} then showed that
\(\dirProblem\) is \(\NP\)-complete, even on dBGs with \(|\Sigma|=4\).
We settle the negative landscape by showing the following result: 

\begin{restatable*}{theorem}{thmDirHardDebruijnInterval}
The $\dirProblem$ problem is $\NP$-complete, even on de Bruijn graphs with $|\Sigma|=2$.
\label{thm:dir-hard-debruijn-interval}
\end{restatable*}

Beyond theoretically interesting, \Cref{thm:dir-hard-debruijn-interval} shows that \emph{in general} it is indeed hard to reconstruct a private (binary) string from a given dBG, thus providing theoretical justification for the privacy framework introduced by Bernardini et al.~\cite{DBLP:conf/alenex/0001CFLP20}.
Given these negative results, we shift our focus to \emph{parametrization} aiming to capture the quality of our domain knowledge in the complexity. This leads to algorithms that are efficient if the intervals are small~\cite{DBLP:journals/bioinformatics/HannenhalliFLSP96}: every interval length is bounded by a natural number $w$, which we term the \emph{interval width}. We denote this parametrized version of $\dirProblem$ by $\fptProblem$.
Hannenhalli et al.~\cite{DBLP:journals/bioinformatics/HannenhalliFLSP96} showed an algorithm for $\fptProblem$ working in $\BigO(m^{2w+\log(2w)})$ time for any directed graph $G=(V,E)$, with $m=|E|$ and $w=\BigO(1)$. Ben{-}Dor et al.~\cite{DBLP:journals/jcb/Ben-DorPSS02} developed an algorithm to solve the problem on dBGs in $\BigO(m \cdot \iWidth^{1.5} 4^{\iWidth})$ time, for any $w$. Bumpus and Meeks~\cite{DBLP:journals/algorithmica/BumpusM23} later rediscovered the same algorithm on temporal graphs, which highlights the relevance of the \(\dirProblem\) problem in other contexts. We may thus summarize the state of the art in the
following statement. 

\begin{restatable*}[Theorem 11 of~\cite{DBLP:journals/jcb/Ben-DorPSS02}, Theorem 7 of~\cite{DBLP:journals/algorithmica/BumpusM23}]{theorem}{thmFpt}
\label{thm:fpt}
There is an $\BigO(m \cdot \iWidth^{1.5} 4^{\iWidth})$-time algorithm solving the $\fptProblem$ problem.
Therefore, $\fptProblem$ is in $\fpt$.
\end{restatable*}

We observe that the state-of-the-art algorithm does not exploit the dBG structure.
Our central contribution is developing an algorithm for dBGs parametrized by $w (\log w+1) /(k-1)$:

\begin{restatable*}{theorem}{thmFptDbg}
\label{thm:fpt-dbg}
Let \(G\) be a de Bruijn graph of order \(k\) over alphabet \(\Sigma\), \(|\Sigma|=\BigO(1)\). There is an $\BigO(m \cdot \lambda^{\frac{\iWidth}{k-1}+1})$-time algorithm solving the $\fptProblem$ problem, where \(\fptBase \coloneqq \min(\setsize{\Sigma}^{k-1}, 2w-1)\).
\end{restatable*}

\Cref{thm:fpt-dbg} improves on the state of the art by roughly an exponent of $(\log w +1)/(k-1)$. The existing algorithms~\cite{DBLP:journals/bioinformatics/HannenhalliFLSP96,DBLP:journals/jcb/Ben-DorPSS02,DBLP:journals/algorithmica/BumpusM23} all have a natural interpretation for string reconstruction: when for each length-$k$ substring (\emph{$k$-mer}), we know a small range of positions it must lie in, string reconstruction can be solved in linear time. \Cref{thm:fpt-dbg} shows that \emph{it is enough when the range of positions is small relative to the order $k$} of the dBG. 
In particular, we show that for dBGs it is sufficient if $w \log w/(k-1)$ is relatively small, which significantly extends the practical applicability of our technique.
For instance, in bioinformatics, it is standard to use $k=31$~\cite{DBLP:journals/bioinformatics/LimassetFP20} and then we have $|\Sigma|=4$ (the size of the DNA alphabet), which implies an exponential speedup by $\sqrt[30]{\cdot}$. Our approach of improving the $\fpt$ algorithms for \(\dirProblem\) based on combinatorial insights into the structure of the instances suggests further research into closely-related problems; e.g., the \emph{Hierarchical Chinese Postman} (HCP) problem~\cite{DBLP:journals/orl/GhianiI00,DBLP:journals/eor/KortewegV06,DBLP:journals/networks/DrorST87,DBLP:journals/orl/AfanasevBT21} and the related \emph{Time-Constrained Chinese Postman} (TCCP) problem~\cite{WANG2002375,6025623}.

We then generalize the above results by allowing the cost at every position of an interval to vary. We denote this problem here by \(\dirProblemCost\) (see \Cref{sec:setting} for a formal definition).
In this setting, our hardness result (\Cref{thm:dir-hard-debruijn-interval}) translates into inapproximability. 

\begin{restatable*}{corollary}{corDebruijnInapprox}\label{cor:inapprox} 
    If \(\PClass \neq \NP\) there is no constant-factor polynomial-time approximation algorithm for the \(\dirProblemCost\) problem, even on de Bruijn graphs with interval cost functions.
\end{restatable*}

We show that the $\fpt$ algorithms underlying both \Cref{thm:fpt,thm:fpt-dbg} also work in the optimization version with an interval cost function, which we denote by \(\fptProblemCost\).

\begin{restatable*}{corollary}{findingAlgGen}
\label{thm:finding-alg-gen}
Given a \(\fptProblemCost\) instance, there is an $\BigO(m \cdot \iWidth^{1.5} 4^{\iWidth})$-time algorithm finding a min-cost Eulerian trail in \(G\).
On a de Bruijn graph of order \(k\) over alphabet \(\Sigma\), \(|\Sigma|=\BigO(1)\), we can solve this problem in $\BigO(m \cdot \lambda^{\frac{\iWidth}{k-1}+1})$ time, where \(\fptBase \coloneqq \min(\setsize{\Sigma}^{k-1}, 2w-1)\).
\end{restatable*}

Surprisingly, even in this more general setting, we show how to extend our \(\fpt\) techniques to count the number of min-cost Eulerian trails. We show the following result.

\begin{restatable*}{theorem}{thmAlgCounting}
    \label{thm:alg-counting}
    Given a \(\fptProblemCost\) instance, we can count the number of min-cost Eulerian trails in \(\BigO(m\cdot w^{1.5} 4^w)\) time.  
    On a de Bruijn graph of order \(k\) over alphabet \(\Sigma\), \(|\Sigma|=\BigO(1)\), we can solve this problem in $\BigO(m \cdot \lambda^{\frac{\iWidth}{k-1}+1})$ time, where \(\fptBase \coloneqq \min(\setsize{\Sigma}^{k-1}, 2w-1)\).
\end{restatable*} 

We can also enumerate these trails in the same time as for counting (\Cref{thm:alg-counting}) plus time that is linear in the size of the output.
Notably, all of our algorithmic results translate from directed to undirected graphs with the same complexities. It is easy to verify that none of the proofs depend on the directedness of the graph.
For simplicity, we focus our discussion on directed graphs.
Particularly our result for counting Eulerian trails in undirected graphs is surprising given that the problem is \(\SharpP\)-complete in the standard setting~\cite{DBLP:conf/alenex/BrightwellW05}. We also show that most of our algorithmic results generalize to multigraphs.
Finally, we show that the undirected version of $\dirProblemCost$, which we denote by $\undirProblemCost$, is also $\NP$-complete. 

\subparagraph{Overview of Hardness Techniques.}
We consider dBGs over alphabets of size two.
Unlike~\cite{DBLP:journals/bioinformatics/HannenhalliFLSP96,DBLP:journals/jcb/Ben-DorPSS02},
this highly structured setting requires intricate techniques to obtain an (elementary) reduction from the directed Hamiltonian path problem~\cite{DBLP:conf/coco/Karp72}.
For our instances, we harness the fact that the shortest path between any two nodes in a complete dBG is unique and has a meaningful string interpretation.
The structure of the reduction makes it directly translatable to the optimization setting implying inapproximability.
Similarly, we reduce the undirected Hamiltonian path problem~\cite{DBLP:conf/coco/Karp72} to finding a min-cost $c$-respecting Eulerian trail in an undirected graph.
This reduction requires a different approach than the directed case to ensure that critical parts of the graph can only be traversed in the desired order.
We assign different costs at even and odd time steps to certain edges to achieve that goal.
\subparagraph{Overview of Algorithmic Techniques.}
We apply tools from parametrized algorithms~\cite{DBLP:books/sp/CyganFKLMPPS15} to solve \(\fptProblemCost\) efficiently.
Our techniques can be viewed as a careful combination of searching in a well-bounded state space and dynamic programming. We further combine approaches from graph and string algorithms to enhance the vanilla version of these tools with combinatorial insights into dBGs and obtain \underline{exponential-time improvements}. The robustness of these tools allows us to generalize our algorithm for the decision version of \(\fptProblemCost\) to both the optimization and the counting versions. The counting result in particular relies on the fact that the state space compactly captures all possible Eulerian trails, by excluding impossible trails at the construction level, and representing the possible ones efficiently.

\subparagraph{Other Related Work.} Our work is closely related to exploring \emph{temporal graphs}~\cite{DBLP:journals/tcs/MichailS16,DBLP:conf/mfcs/ErlebachS18,DBLP:conf/sirocco/ErlebachS20,DBLP:journals/jcss/Erlebach0K21,DBLP:journals/jcss/AkridaMSR21,DBLP:journals/algorithmica/BumpusM23,DBLP:journals/algorithmica/MarinoS23}: graphs where every edge is available at an arbitrary subset of the time steps. 
Most relevant to our work are perhaps the works by Bumpus and Meeks~\cite{DBLP:journals/algorithmica/BumpusM23} and by Marino and Silva~\cite{DBLP:journals/algorithmica/MarinoS23}.
The former proved that deciding whether a temporal graph has an Eulerian trail is $\NP$-complete even if each edge appears at most $r$ times, for every fixed $r\geq 3$.
They also apply similar parametrized tools to interval-membership-width, a temporal graph parameter related to our interval width, but they do not consider directed graphs, costs, or counting.
Marino and Silva~\cite{DBLP:journals/algorithmica/MarinoS23} showed, for undirected graphs, that, if the edges of a temporal graph $(G,\lambda)$ with lifetime $\tau$ are always available during the lifetime, then deciding whether $(G,\lambda)$ has an Eulerian trail is $\NP$-complete even if $\tau=2$. They also showed that this problem is in \textsf{XP} when parametrized by $\tau+\textsf{tw}(G)$, where $\textsf{tw}(G)$ is the treewidth of $G$. 

Our work is also closely related to the HCP problem~\cite{DBLP:journals/orl/GhianiI00,DBLP:journals/eor/KortewegV06,DBLP:journals/networks/DrorST87,DBLP:journals/orl/AfanasevBT21} (and the related TCCP problem~\cite{WANG2002375,6025623}). In HCP, we are given an edge-weighted undirected graph $G = (V, E)$, a partition $\mathcal{P}$ of $E$ into $k$ classes, and a partial order $\prec$ on $\mathcal{P}$, and we are asked to find a least-weight closed walk traversing each edge in $E$ at
least once such that each edge $e$ in a class $E'$ is traversed
only after all edges in all classes $E'' \prec E'$ are traversed.
The main differences to our problem are two: (1) In HCP, there is a partition on the edges and then a partial order between the classes, whereas we have an arbitrary interval per edge; and (2) in HCP, one must traverse each edge at least once, whereas we have to do this exactly once. 
\subparagraph{Paper Organization.} 
We begin by formalizing $\dirProblem$ and $\dirProblemCost$ in \Cref{sec:setting}.
We present our hardness and inapproximability results in \Cref{sec:hardness}.
We complement these negative results with positive algorithmic results for $\fptProblem$ in \Cref{sec:fpt}.
In \Cref{sec:counting}, we extend our algorithms to counting Eulerian trails
and discuss the relationship to the data privacy applications.
We also show slight extensions to the algorithms of Ben-Dor et al.~and Bumpus and Meeks in \Cref{app:generalizations}.
Finally, we discuss how most of our algorithms extend to multigraphs in \Cref{sec:multigraphs}.

\section{Preliminaries}\label{sec:setting}
For $a, b \in \N$, let $[a, b] \coloneqq \{x \in \N \mid a \le x \le b\}$, $[a, b) \coloneqq \{x \in \N \mid a \le x < b\}$, and $[b] \coloneqq [1,b]$.
Further let \(\intervals{b} \coloneqq \{[i,j] \mid 1 \le i \le j \le b\} \cup \{\{\}\}\) be the set of intervals on \([b]\).

We imagine a trail as ``walking'' through a graph and thus refer to the ranks of the edges as happening at certain time steps (e.g., the second edge is associated with the time step 2).
This vocabulary borrowed from temporal graphs makes the discussion more intuitive.

\begin{definition}
    \label{def:interval-cost}
    Let \(G = (V,E)\) be a directed or undirected graph with $m=|E|$.
    We call a function $c \colon E  \to \intervals{m}$ an \emph{interval function} for $G$.
    We extend this notion to an \emph{interval cost function} where, for each \(e \in E\), each time step from the interval \(c(e)\) is associated with a cost from \(\Z\). 
    As an abuse of notation, we write \(c(e,t)\) for this cost. If $t \not \in c(e)$, we set \(c(e,t) \coloneqq \infty\).  
\end{definition}

When the relevant graph $G=(V,E)$ is clear from context, we set $n\coloneqq|V|$ and $m\coloneqq|E|$. We next formalize the main problems in scope on directed graphs.
\begin{problem}{Eulerian Trails in Digraphs with Interval functions ($\dirProblem$)}
    Given: & Directed graph $G=(V, E)$, interval function $c \colon E  \to \intervals{m}$ \\   
    Decide: & Is there an Eulerian trail $P = e_1\dots e_m$ in $G$ such that for all $t \in [m]$, $t \in c(e_t)$? 
\end{problem} 

Similarly, we consider the version with interval \emph{cost} functions. 

\begin{problem}{Eulerian Trails in Digraphs with Interval Cost functions ($\dirProblemCost$)}
    Given: & Directed graph $G=(V, E)$, interval cost function $c \colon E \times [m]  \to \Z \cup \{\infty\}$, $\cmax \in \N$ \\
    Decide: & Is there an Eulerian trail $P = e_1\dots e_{m}$ in $G$ such that $\sum_{t = 1}^{m} c(e_t,t) \le \cmax$?
\end{problem} 
In addition, we consider the \(\undirProblemCost\) problem as the natural \emph{undirected} version of $\dirProblemCost$.

Note that in all cases, we allow setting every interval to the whole \([m]\). 
In the case without costs, this yields the classical Eulerian trail problem~\cite{hierholzer1873moglichkeit}.
In the case with costs, this yields a min-cost Eulerian trail problem without any interval restrictions, which is a natural special case we will call out when considering it to be particularly relevant.
See \Cref{fig:available-example} for an example of a graph with an interval function and the corresponding edge availabilities.

\begin{figure}[tbhp]
    \centering
    \includegraphics[width=0.65\textwidth, page=1]{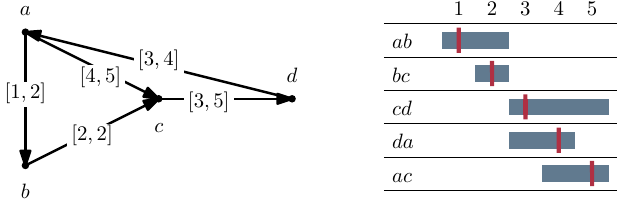}
    \caption{On the left is the input graph $G$: every edge is labeled with the time steps at which it is available. On the right is a table illustrating the interval every edge is available: $abcdac$ is an Eulerian trail; the edge usages corresponding to this trail are indicated with red vertical lines.}
    \label{fig:available-example}
\end{figure}%
The following definition makes our arguments based on an interval function more legible.
\begin{definition}
    \label{def:available}
    We call an edge $e \in E$ \emph{available} at $t \in [m]$, if $t \in c(e)$.
    For a time step $t \in [m]$, we write $E(t)$ to denote the set of edges available at $t$.
    Given an interval $T \subseteq [m]$, we similarly write $E(T)$ for the set of edges available at at least one time step from $T$.
\end{definition}

\section{NP-hardness and Inapproximability}
\label{sec:hardness}
In this section, we consider both the directed and undirected case from the hardness perspective. Since in the directed case, dBGs are a particularly relevant application, we give strong results, even when restricting the instances to this graph class.
We are the first to show the hardness for alphabets of all sizes (except $\setsize{\Sigma} = 1$, where the problem is trivial). 
In particular, we settle the question for the important case of binary alphabets; see \Cref{sec:dir-hard}.
In the undirected case, we consider general graphs, which we believe is nonetheless insightful regarding what makes this problem class hard; see \Cref{sec:undir-hard}.

\subsection{Directed Graphs}\label{sec:dir-hard}
We first settle the computational complexity of $\dirProblem$. 
\thmDirHardDebruijnInterval

We reduce from the directed Hamiltonian path problem, which is $\NP$-complete~\cite{DBLP:conf/coco/Karp72}.
Let \(G=(V, E)\) be the directed graph in which we want to find a Hamiltonian path.

\subparagraph{Construction of the dBG.} We consider any two-letter alphabet; for example, here we will consider the subset \(\Sigma \coloneqq \{\letA, \letT\}\) of the DNA alphabet. We set \(\ell \coloneqq \lceil \log_2(\setsize{V}) \rceil\) and
construct a \(\dirProblem\) instance $(G'=(V',E'),c)$ as the complete dBG on the set of nodes each labeled by a string of length $k-1=4\ell + 10$, 
as well as an interval function $c$ that we will describe later. 
Note that, by the choice of $\ell$ and $\Sigma$, $G'$ has $|\Sigma|^{k-1}=2^{4\ell + 10} = \BigO(\setsize{V}^4)$ nodes. We will pay special attention to some of the nodes in $G'$, which we assign an interpretation in terms of $G$.

For each \(v \in V\), let \(\id v \in \Sigma^\ell\) be a unique identifier among \(V\). By the choice of \(\ell\), this numbering is always possible.
Similarly, let \(\id{v'} \in \Sigma^{4\ell+10} \) be the string associated with a node $v' \in V'$.
With each node \(v \in V\), we associate two nodes \(v'_1\) and \(v'_2\) in $V'$ such that: 
\begin{align*}
    \id{v'_1} &\coloneqq \letA^{\ell + 3} \cdot \letT \cdot \id{v} \cdot \letT \cdot \letA^{\ell+3} \cdot \letT \cdot \id{v} \cdot \letT \text{ and }
    \id{v'_2} \coloneqq \letA^{\ell + 2} \cdot \letT \cdot \id{v} \cdot \letT \cdot \letA^{\ell+3} \cdot \letT \cdot \id{v} \cdot \letT \cdot \letA,
\end{align*}
where $\letA^{i}$ is the string of $i$ $\letA$'s.
We use the notation $a \overset{x}{\to} b$ to denote a directed edge $ab$ labeled $x$, where $x$ is the letter used to obtain $b$ from $a$ in the dBG. 
We extend this notation to
$a\overset{X}{\rightsquigarrow}b$ to denote the path from $a$ to $b$ using string $X$.
Observe that by the choice of the underlying strings, we have an edge \(v'_1 \xrightarrow{\letA} v'_2\) in the dBG.
We call this the \emph{inner edge} for \(v\).
With each edge \(e = vu \in E\), we associate two nodes \(e'_1\) and \(e'_2\) such that:
\begin{align*}
    \id{e'_1} &\coloneqq \letA^{\ell + 3} \cdot \letT \cdot \id{v} \cdot \letT \cdot \letA^{\ell+3} \cdot \letT \cdot \id{u} \cdot \letT \text{ and }
    \id{e'_2} \coloneqq \letA^{\ell + 2} \cdot \letT \cdot \id{v} \cdot \letT \cdot \letA^{\ell+3} \cdot \letT \cdot \id{u} \cdot \letT \cdot \letA.
\end{align*}
Again, we have an edge \(e'_1 \xrightarrow{\letA} e'_2\) in the dBG. We call this the \emph{inner edge} for \(e\).
For every other string of length $4\ell+10$ over $\Sigma$ a node $v'$ exists in $V'$, but it is not associated with a node $v\in V$ or an edge $e\in E$.
See \Cref{fig:hardness} for an illustration of this construction.
A crucial property that we will later utilize is that for every pair of nodes $v,u \in V$, such that $vu\in E$, there is a unique shortest path of length $2\ell + 4$ in $G'$ between $v_2'$ and $(vu)_1'$; this path is associated with the string $\letA^{\ell + 2} \cdot \letT \cdot \id{u} \cdot \letT$ (see \Cref{fig:hardness}). The analogous property holds for $(vu)_2'$ and $u'_1$: there is a unique shortest path of length $2\ell + 4$ in $G'$ from $(vu)_2'$ to $u'_1$.

\begin{figure}[tbhp]
    \centering
    \includegraphics[width=0.8\linewidth,page=5]{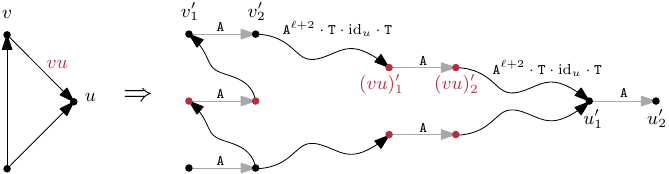}
    \caption{A directed graph $G$ (left) and the corresponding part of graph $G'$ (right). The nodes in $G'$ that correspond to nodes in $G$ are colored black; the nodes in $G'$ that correspond to edges in $G$ are colored red. Inner edges in $G'$ are colored gray and squiggly lines indicate \emph{paths} of length $2\ell +4$.}
    \label{fig:hardness}
\end{figure}

It is straightforward to verify that the claimed path (from \(v_2'\) to \((vu)_1'\)) exists.
To see that this is the shortest possible, notice that \(\id{v_2'}\) contains the substring \(\letA^{\ell + 3}\) starting at position \(2 \ell + 5\) (i.e., after the first occurrence of \(\id{v} \cdot \letT\)).
Any one edge on a path to \((vu)_1'\) shifts this substring left by one letter.
By the strategic placement of \(\letT\) around the \(\id{}\) parts, the substring \(\letA^{\ell +3}\) occurs precisely at positions \(1\) and \(2\ell + 6\) in \(\id{(vu)_1'}\).
Thus, the closest occurrence of \(\letA^{\ell +3}\) left of position \(2\ell + 6\) in \(\id{(vu)_1'}\) starts at position 1 (i.e., requires a left shift by at least \(2\ell +4\) letters). Therefore, any path from \(v_2'\) to \((vu)_1'\) must contain at least \(2 \ell +4\) edges. The argument is analogous for the path from $(vu)_2'$ to $u'_1$.

\subparagraph{Construction of the Interval Function.}
Let us set \(\tau \coloneqq 2n-1+2(n-1)(2\ell+4)\).
We dub the time steps 1 to \(\tau\) \emph{early} and all others \emph{late}.
The intuition is that the first \(\tau\) edges on any $c$-respecting Eulerian trail in \(G'\) will correspond to a Hamiltonian path in \(G\).
For ease of notation, we set \(m'\coloneqq|\Sigma|^{k}=2^{4\ell + 11}\) to be the number of edges in \(G'\).

For the interval function $c$, we set \(c(v'_1 v_2') \coloneqq [1,\tau]\), for all $v\in V$. For any non-edge \(e \not \in E\) we set \(c(e_1' e'_2) \coloneqq [\tau+1, m']\). All other edges in $G'$ have the full interval \([m']\).
This also holds for all inner edges for the edges in $E$. Note that with these choices, the inner edges for nodes \(v \in V\) are only available early, and for a node pair $v,u \in V$, the edge $(vu)'_1 \to (vu)'_2$ is available early only if $v \to u$ is an edge in \(E\).

\begin{lemma}
\label{lem:debruijn-dhp-to-et}
If \(G\) has a Hamiltonian path, there is a $c$-respecting Eulerian trail in \((G', c)\).
\end{lemma}
\begin{proof}
Let \(P = v_1 \dots v_n\) be a Hamiltonian path in \(G\). Then for the Eulerian trail in \(G'\), set
$$
P'_1 \coloneqq (v_1)'_1 \overset{\letA}{\to} (v_1)'_2 \overset{\letA^{\ell + 2} \cdot \letT \cdot \id{v_2} \cdot \letT}{\rightsquigarrow} (v_1v_2)_1' \overset{\letA}{\to} (v_1v_2)_2' \rightsquigarrow \dots \rightsquigarrow (v_n)_1' \overset{\letA}{\to} (v_n)_2'.
$$
Notice that \(P_1'\) traverses all inner edges for nodes and that \(P'_1\) has length exactly \(\tau\).
Thus all inner edges for nodes are available by construction at their time steps on \(P'_1\).
Similarly, for any \(i \in [n-1]\), we have that \(v_iv_{i+1} \in E\) since \(P\) is a Hamiltonian path in $G$. Thus, the inner edges for edges are available on \(P'_1\) as well. Finally, all other edges on \(P'_1\) have their $\letA^{\ell+3}$ substring in such positions that they do not fall under any of the edge classes whose interval is restricted, thus they are available for the full interval $[m']$. We conclude that \(P'_1\) respects $c$.

As we want \(P'_1\) to be the prefix of an Eulerian trail, we must ensure that no edge repeats in \(P'_1\).
The nodes of types $v_1'$, $v_2'$, $e_1'$, and $e_2'$ in \(G'\), for some $v \in V$ or $e \in E$, each appears at most once on $P_1'$ by construction, and thus so must their inner edges.
Therefore, if there is an edge $xy$ in \(G'\), for two nodes $x$ and $y$, that appears twice on $P_1'$, it must be on some subpath of type $(v_i)'_2 \overset{\letA^{\ell + 2} \cdot \letT \cdot \id{v_{i+i}} \cdot \letT}{\rightsquigarrow} (v_iv_{i+1})'_1$ (or analogously of the type $(v_iv_{i+1})'_2 \overset{\letA^{\ell + 2} \cdot \letT \cdot \id{v_{i+1}} \cdot \letT}{\rightsquigarrow} (v_{i+1})'_1$) as well as at some other part of $P_1'$.
Let us analyze the string corresponding to node $x$. As $x$ appears on the path from $(v_i)'_2$  to $(v_iv_{i+1})'_1$, it must have the form $r_1 \cdot \letA^{\ell + 3} \cdot \letT \cdot \id{v_i} \cdot \letT \cdot r_2$, where $r_1$ is a suffix of $\letA^{\ell + 2} \cdot \letT \cdot \id{v_i} \cdot \letT$ and $r_2$ is a prefix of $\letA^{\ell+3} \cdot \letT \cdot \id{v_{i+1}} \cdot \letT$. Consider that therefore the string corresponding to $x$ contains the substring $\letA^{\ell + 3} \cdot \letT \cdot \id{v_i} \cdot \letT$ somewhere in the middle.
Since, however, our id's are unique, $x$ may only appear on the subpath corresponding to an out-edge of $v_i$ in $G$ and a Hamiltonian path can only include one out-edge per node. Finally, no edge can repeat within the same such subpath as this is a shortest path. The above discussion contradicts our assumption that $xy$ repeats on $P_1'$. Therefore, we conclude that all edges on $P_1'$ are unique.

We are left to show that we can complete \(P_1'\) to traverse every edge of $G'$ exactly once.
Observe that since \(G'\) is a complete dBG, the in- and out-degree of every node is \(\setsize{\Sigma} = 2\).
Therefore, in \(G' - P_1'\) (which we define as $G'$ without the edges in $P_1'$) the in- and out-degree of every node is the same except for \((v_1)_1'\) (where the out-degree is one less than the in-degree) and for \((v_n)_2'\) (where the out-degree is one more than the in-degree). 
Also notice that $G'-P'_1$ remains weakly connected (as any strongly connected $2$-regular directed graph remains weakly connected if a path is removed). By these properties, there is an Eulerian trail \(P'_2\) in \(G'-P'_1\). In particular, it must start at $(v_n)_2'$ and end at $(v_1)_1'$. Hence $P' \coloneqq P_1' P_2'$ is an Eulerian trail in \(G'\).
As the time steps of all edges from \(P_1'\) remain the same, that prefix of $P'$ still respects $c$.
For the suffix of $P'$ defined by \(P'_2\), observe that all edges are used at time steps later than \(\tau\) and that none of the edges are inner edges for nodes.
Thus, all edges in the \(P_2'\) part of \(P'\) are also available under $c$.
We conclude that \(P'\) is an Eulerian trail in \((G', c)\).
\end{proof}

Similar ideas allow us to prove the other direction, yielding the following lemma.

\begin{lemma}
\label{lem:debruijn-et-to-dhp}
    If there is a $c$-respecting Eulerian trail in \(G'\), \(G\) has a Hamiltonian path.
\end{lemma}
\begin{proof}
Let \(P'\) be a $c$-respecting Eulerian trail in \(G'\).
If \(P'\) has the exact structure as in the proof of \Cref{lem:debruijn-dhp-to-et}, the result is immediate.
As we cannot (without further proof) make such assumptions upon the structure of \(P'\), this direction needs a slightly more careful argument.
Yet, by our choice of the interval function we obtain the needed structural properties.

First, let \(v_1, \dots, v_n\) be the nodes from \(V\) in the order that the edges \((v_i)_1' \to (v_i)_2'\) appear on \(P'\).
Since \(P'\) respects $c$, all of these inner edges for nodes must have time steps at most \(\tau\).
For any \(i \in [n-1]\), there must be a subtrail from \((v_i)_2'\) to \((v_{i+1})_1'\).
By looking at their strings, we see that the longest suffix of \(\id{(v_i)_2'}\) that is a prefix of \(\id{(v_{i+1})_1'}\), is just $\letA$.
Thus, the path \[(v_i)_2' \overset{\letA^{\ell + 2} \cdot \letT \cdot \id{v_{i+1}} \cdot \letT}{\rightsquigarrow} (v_iv_{i+1})_1' \overset{\letA}{\to} (v_iv_{i+1})_2' \overset{\letA^{\ell + 2} \cdot \letT \cdot \id{v_{i+1}} \cdot \letT}{\rightsquigarrow} (v_{i+1})'_1\] is the unique shortest path between these two nodes.
This path has length \(2(2\ell+4)+1\).
As there are \(n\) edges of type \((v_i)_1' \to (v_i)_2'\), we must have \(n-1\) such connecting paths between them.
As \(n + (n-1)\cdot(2(2\ell+4)+1) = \tau\), we just have enough time to traverse \((v_n)_1' \to (v_n)_2'\) in time (that is before time step \(\tau\)), if we assume the intermediate subtrails as the shortest possible, so since \(P'\) respects $c$, that must be precisely the structure of the length-$\tau$ prefix of \(P'\).
Therefore, the inner edges for nodes appear on \(P'\) connected to each other by a path containing  an inner edge for edges.
Since such inner edges for (potential) edges are available only if the corresponding edge in \(G\) exists, for all \(i \in \), we have that \(v_iv_{i+1} \in E\).
Thus, \(P \coloneqq v_1 \dots v_n\) is a Hamiltonian path in \(G\), as desired.
\end{proof}

\Cref{lem:debruijn-dhp-to-et} and \Cref{lem:debruijn-et-to-dhp} allow us to deduce \Cref{thm:dir-hard-debruijn-interval}.

\thmDirHardDebruijnInterval*
\begin{proof}
As each node in $G'$ is associated with a string of length $4\lceil \log \setsize{V} \rceil + 10$, it has at most $2^{4 \log \setsize{V} + 11} = 2048 \setsize{V}^4$ nodes.
Thus, the construction of $G'$ is clearly possible in polynomial time, and
the statement follows from \Cref{lem:debruijn-dhp-to-et,lem:debruijn-et-to-dhp}.

To see that the problem is in \(\NP\) consider the canonical certificate that encodes the desired trail as a sequence of nodes.
\end{proof}

\subparagraph{Consequences.} The construction we used is quite restrictive in that it uses only two letters to construct the dBG and is within the framework of interval functions without costs.
This allows us to deduce the following two corollaries, narrowing down the hardness landscape of the general problem even further. \Cref{cor:2-regular} was first shown by Hannenhalli et al.~\cite{DBLP:journals/bioinformatics/HannenhalliFLSP96}.

\begin{corollary}[\cite{DBLP:journals/bioinformatics/HannenhalliFLSP96}]\label{cor:2-regular}
    The $\dirProblem$ problem remains $\NP$-hard even on 2-regular graphs.
\end{corollary}

\begin{proof}
    Notice that the construction for \Cref{thm:dir-hard-debruijn-interval} uses only the letters \(\{\letA, \letT\}\) and thus the complete dBG (over that alphabet) is 2-regular.
\end{proof}

Our hardness reduction also directly implies that \(\dirProblemCost\) is $\NP$-hard: we set the costs of every edge to \(0\) within its interval and to $1$ outside its interval.
We also set \(\cmax \coloneqq 0\).
As our proof of \Cref{thm:dir-hard-debruijn-interval} relies only on distinguishing zero and non-zero cost instances, we deduce the following inapproximability result, which is to the best of our knowledge new.

\corDebruijnInapprox

\begin{proof}
    Assume that there is an approximation algorithm that solves this variant of the problem with a constant factor \(c \in \R^+\).
    Now, clearly, if there is a zero-cost Eulerian trail in the constructed instance, the approximation algorithm must also output zero.
    And conversely, if there is no such zero-cost Eulerian trail then the approximation algorithm must output some non-zero answer.
    In this way, we can employ the approximation algorithm to solve the decision problem, implying that it does not exist unless \(\PClass = \NP\).
\end{proof}

Interestingly, the problem remains hard and inapproximable, even when restricting \(\dirProblemCost\) to instances where every interval is the full interval \([m]\) (effectively dropping the interval requirement and only considering edge costs).
This reduction simply extends, for each edge, its interval and let all costs outside the original interval be \(\cmax + 1\). 

Although we can find an Eulerian trail in an undirected graph in linear time by Hierholzer's algorithm~\cite{hierholzer1873moglichkeit}, counting the number of Eulerian trails in an undirected graph is \(\SharpP\)-complete~\cite{DBLP:conf/alenex/BrightwellW05}.
Given this complexity landscape and our results on directed graphs, it is natural to ask about the complexity of deciding if an Eulerian trail of at most some fixed cost exists.
We show that this problem is \(\NP\)-complete and inapproximable. 

\subsection{Undirected Graphs}\label{sec:undir-hard}

In this section, we show that $\undirProblemCost$ is \(\NP\)-complete and inapproximable.
We reduce from the undirected Hamiltonian path problem, which is famously known to be $\NP$-complete~\cite{DBLP:conf/coco/Karp72}.
Let $G = (V,E)$ be any undirected graph and let $n = |V| \ge 3$.
We will construct a $\undirProblemCost$ instance $(G' = (V', E'), c)$ as follows.
Let $V' = (V\times \{1,2\}) \cup \{x\}$ where $x\not \in V$ is a new node.
Thus we duplicate every node $v$ of $G$ into two nodes (which we denote by $v_1'$ and $v_2'$). 
Let $E' = {V' \choose 2}$, that is, let $G'$ be the complete graph on $V'$.
Set \(m' \coloneqq \setsize{E'}\).

As stated in \Cref{thm:undir-hard}, we show the result even when all intervals are complete (i.e., \([m']\)). For the cost, intuitively, we categorize every edge as \emph{cheap} (cost zero) or \emph{expensive} (cost non-zero) based on two factors: (1) whether we want it to be taken as part of the Hamiltonian path; and (2) whether we want the edge to be taken at an even or odd timestamp.

In our construction, a Hamiltonian path in $G$ corresponds to a sequence of cheap edges in $G'$ that alternates between taking \emph{inner} edges (i.e., ones between $v'_1$ and $v'_2)$ and \emph{external} edges (i.e., ones connecting different nodes from $G)$. In $G'$, this cheap path can then be extended to a cheap Eulerian trail.
The construction will similarly ensure that every cheap Eulerian trail in $G'$ must have a prefix that corresponds to a Hamiltonian path in $G$. For $v,w \in V$, $i,j \in \{1,2\}$ and $t \in [|E|]$, an edge $\{v_i', w_j'\}$ has cost zero at $t$ if and only if:
\begin{enumerate}
    \item it is an inner edge (i.e., $v=w$), $t$ is even and $t<2n+1$; 
    \item there is an edge $\{v,w\} \in E$, $t$ is uneven and $t \in [3, 2n+1)$; or
    \item it is not an inner edge (i.e., $v\ne w$) and $t > 2n+1$.
\end{enumerate}
Regarding the special node $x$, set all adjacent edges to always have cost zero, that is $c(\{x, v_i'\}, t) \coloneqq 0$.
All other costs are set to 1.
The construction from a Hamiltonian path to an Eulerian trail is illustrated in \Cref{fig:ucET-construction}.

\begin{figure}[hbtp]
    \centering
    \includegraphics[width=0.5\textwidth]{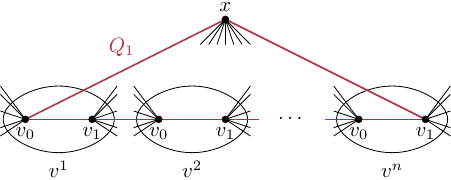}
    \caption{The zero-cost cycle $Q_1 \subseteq E'$ constructed from a Hamiltonian path in $G$.}
    \label{fig:ucET-construction}
\end{figure}

\begin{theorem}
\label{thm:undir-hard}
The $\undirProblemCost$ problem is \(\NP\)-complete, even when all edges are always available.
\end{theorem}

\begin{proof}
To prove the validity of this reduction, first assume that $G$ has a Hamiltonian path $v^1v^2\dots v^n$.
Then
$$
Q_1 \coloneqq x v'^1_1 v'^1_1 v'^2_1 v'^2_2 \dots v'^n_1 v'^n_2 x
$$
is a cycle in $G'$ of cost zero as all edges are all on or before timestamp $2n+1$ and are
(1) inner edges for a node from $V$ and at an even timestamp;
(2) between nodes corresponding to two different nodes in $G$ sharing an edge and at an odd timestamp; or
(3) adjacent to $x$.

Observe that $G'$ is Eulerian as it is the complete graph on an odd number of nodes.
Also, $G' - Q_1$ remains connected as every node, except $v'^1_1$ and $v'^n_2$, has an edge to $x$.
Since $n \ge 3$, the two remaining nodes have at least two edges in $G' - Q_1$, thus at least one to a node that has an edge to $x$.
Thus $G' - Q_1$ is still Eulerian.
Let $Q_2$ be an Eulerian cycle in $G' - Q_1$.
Without loss of generality, assume that $Q_2$  starts at $x$.
Then $Q \coloneqq Q_1 Q_2$ is an Eulerian trail in $G'$.
In $Q$, all edges from $Q_2$ have timestamp at least $2n+2$ and none of them are inner edges (as all of the inner edges are in $Q_1$).
Thus, $Q$ has a total cost of $0$ and we have $(G', c, 0) \in\undirProblemCost$ as desired.

For the other direction, now assume $(G', c,0) \in \undirProblemCost$.
Thus, there is an Eulerian trail $Q$ in $G'$ of cost $0$.
Therefore, all edges in $Q$ must have cost $0$.
Let $Q_1$ be the first $2n+1$ edges of $Q$.
By definition of $c$ the first and last edge in $Q_1$ must be adjacent to $x$ (otherwise they would not have cost 0).
As the second and $2n$-th edge must be an inner inner edge to have cost $0$, we can deduce that the first and last node in $Q_1$ are $x$.
Similarly, in order for the $n$ inner edges to have cost zero, they must occupy all even timestamps before $2n+1$, thus there are $n$ inner edges $Q_1$.
Every edge of odd timestamp must be an edge between nodes corresponding to different nodes in $G$ that share an edge in $G$.
Thus, $Q_1$ must have the form
$$
Q_1 = x v'^1_\star v'^1_\star \dots v'^n_\star v'^n_\star x,
$$
where $\star$ may be $1$ or $2$.
Also since there must be $n$  different inner edges in $Q_1$, for each $v \in V$, the nodes $v_1$ and $v_2$ must appear exactly once and consecutively in $Q_1$.
Therefore, $v^1 \dots v^n$ is a Hamiltonian path in $G$, as desired.

Our reduction shows that $\undirProblemCost$ is $\NP$-complete.
\end{proof}

Note that our reduction implies inapproximability analogously to \Cref{cor:inapprox}.

\begin{corollary}\label{cor:uinapprox} 
    If \(\PClass \neq \NP\) there is no constant-factor polynomial-time approximation algorithm for $\undirProblemCost$, even when all intervals are set to $[m]$.
\end{corollary}

\section{FPT by Interval Width}\label{sec:fpt}

After these strong negative results, it is now natural to ask whether there is a restricted setting in which we can efficiently solve \(\dirProblem\). Interval functions (see \Cref{def:interval-cost}) are a restricted class of functions that are natural to the string reconstruction application.
In this setting, it is equally natural to ask if more precise domain knowledge leads to more efficient algorithms.
Providing tighter intervals for the time step each edge is to be used reduces the combinatorial complexity as there are fewer valid combinations to consider.
This observation also translates to string reconstruction, where we might expect that reconstructing a string is easier when the location of every $k$-mer is narrowed down to a small interval.

\begin{definition}\label{def:width}
Given a \(\dirProblem\) instance \(((V,E),c)\) with interval function $c$, we define the instance's \emph{interval width} $w$ as the length of the longest interval of any edge $e\in E$ under \(c\).
\end{definition}

Interestingly, this concept is closely related to the notion of \emph{interval-membership width} (\(\imw\), in short) studied by Bumpus and Meeks in the context of temporal graphs~\cite{DBLP:journals/algorithmica/BumpusM23}.
This notion asks how many edges in the graph are relevant to the decision at a specific time step $t$ by looking at the set of edges which have been available at or before $t$ \emph{and} will also be available at or after \(t\).
The \(\imw\) is the largest size of such a set for any $t$.
In \Cref{lem:edges-per-interval}, we show that this parameter is related to the length of the longest interval in our setting (\Cref{def:width}).
A special case of this observation also underlies the \(\fpt\) algorithm of Ben{-}Dor et al~\cite{DBLP:journals/jcb/Ben-DorPSS02}.
This relates the length of time any single edge is available to the total number of edges relevant at a given time.
With this observation, we unveil the equivalence of the algorithms by Bumpus and Meeks and by Ben{-}Dor et al.

For dBGs of order $k$, we severely speed up these algorithms by roughly an exponent of $(\log w + 1)/(k-1)$.
In fact, we show that on dBGs, the problem is not only in \(\fpt\) by \(w\), but also by \(w \log w / k\); namely, we show that it suffices if \(w\) is small compared to \(k\) to make the problem tractable. We achieve this by extending the techniques in~\cite{DBLP:journals/jcb/Ben-DorPSS02,DBLP:journals/algorithmica/BumpusM23} with the properties of dBGs and their underlying string interpretation. 
By utilizing a significantly more general setting, Bumpus and Meeks~\cite{DBLP:journals/algorithmica/BumpusM23} have shown that the algorithm by Ben{-}Dor et al.~\cite{DBLP:journals/jcb/Ben-DorPSS02} applies to a significantly broader graph class than dBGs.
We observe that this leaves plenty of dBG-specific structure unexploited to speed up the task. 
Finally, we show how to extend both the existing and our novel approach to find min-cost Eulerian trails.

Considering \Cref{def:width}, the definition of the parametrized problem \(\fptProblem\) is canonical.
We first consider simple graphs, deferring the discussion of multigraphs to the appendix.
\begin{problem}{$\fptProblem$}
    Parameter: & \(w \in \N^+\) \\
       Given: & Directed graph $G=(V, E)$, interval function $c \colon E  \to \intervals{m}$ with interval width \(w\) \\
    Decide: & Is there an Eulerian trail $P = e_1\dots e_{m}$ such that for all \(t \in [m]\), \(t \in c(e_t)\)? 
\end{problem}
The following structural insight (\Cref{lem:edges-per-interval}) generalizes~\cite[Lemma 8]{DBLP:journals/jcb/Ben-DorPSS02} and underpins their algorithm and our reduction to the Bumpus and Meeks version of the \(\fpt\) algorithm: Since any edge is available only at a bounded number of time steps, a similar bound applies to the set of edges available at an individual time step. This insight implies that there are relatively few choices on which edge to select at a specific time step in an Eulerian trail. 

\begin{lemma}
    \label{lem:edges-per-interval}
    If $(G, c, w)$ is a \textsf{YES}-instance of \(\fptProblem\) and $T$ is an interval in $[m]$, then $|E(T)| \leq |T| + 2w - 2$.
\end{lemma}

\begin{proof}
    Suppose $T = [t,t+p-1]$, so $|T| = p$.
    The available intervals for all edges in $E(T)$ will be sub-intervals of $[t-w+1,t+p-1+w-1]$, as each edge is available for at most $w$ time steps.
    This interval has length $(t+p-1+w-1)-(t-w+1)+1=p+2w-2$. As all edges in $E(T)$ are only available within this interval,  we must be able to traverse them within $p + 2w - 2$ time steps, bounding their number.
\end{proof}

In particular, \Cref{lem:edges-per-interval} implies that at any specific time step (i.e., an interval of length 1), there are at most $2w - 1$ edges available.
Thus, we have that any graph with interval width \(w\) has \(\imw\) at most $2w - 1$.
The condition of this lemma is easy to check in time $\BigO(mw)$ by testing if the number of edges available at each time step is at most $2w - 1$. If this is not the case for any time step, we deduce that $(G,c,w)$ is a \textsf{NO}-instance for $\fptProblem$. We will henceforth assume that the condition has been checked and deduce the following result.

\thmFpt

Three remarks are in order regarding the version of this theorem by Bumpus and Meeks ~\cite{DBLP:journals/algorithmica/BumpusM23}.
First, note that this rendering of the theorem requires a small modification, as we are interested in any Eulerian trail, not necessarily a cycle.
This allows us to slash a factor of \(w\).
Second, their algorithm is for undirected graphs, but it naturally translates to directed graphs.
Third, their version of the algorithm does not exclude a small number of irrelevant states and is thus slightly slower (by another factor of $\sqrt{w}$) than the version by Ben{-}Dor et al. 
For the sake of completeness, we provide a proof with these modifications in \Cref{app:generalizations}.

Specializing the parametrized technique of~\cite{DBLP:journals/jcb/Ben-DorPSS02,DBLP:journals/algorithmica/BumpusM23} to dBGs allows us to design an algorithm that achieves a significant speedup. Our improvements rely on exploiting the specific structure in the dBG graph class in which we are particularly interested in finding Eulerian trails.

\thmFptDbg

Since \(\fptBase\) is defined as a minimum, the algorithm underlying \Cref{thm:fpt-dbg} always runs in time \(\BigO(mw \cdot 2^{\frac{w (\log w+1)}{k-1}})\), as shown by the following lemma.

\begin{lemma}
    For any \(w,k \in \N^+\),
    \((2w-1)^{\frac{w}{k-1}+1}=\BigO(w \cdot 2^{\frac{w (\log w+1)}{k-1}})\).
\end{lemma}
\begin{proof}
    Note that \((2w-1)^{\frac{w}{k-1}+1} = (2w-1) \cdot (2w-1)^{\frac{w}{k-1}}\). We first bound the exponential term. Then multiplying by $(2w-1)$ gives the claimed upper bound. We have that
    \begin{align*}
       (2w-1)^{\frac{w}{k-1}}= 2^{\frac{w \log(2w-1)}{k-1}} \le 2^{\frac{w \log(2w)}{k-1}} = 2^{\frac{w (\log w + \log 2)}{k-1}} = 2^{\frac{w (\log w + 1)}{k-1}}.
    \end{align*}
\end{proof}

Notice that the exponential part is higher than for the \(\BigO(mw^{1.5} \cdot 4^w)\)-time algorithm only if $w(\log w+1)/(k-1) >2w$, that is, if \(w = \Omega(4^k)\); but in this case, since \(\fptBase\) is defined by the minimum, our new algorithm takes \(\BigO(mw \cdot \setsize{\Sigma}^w)\) time.
Thus, for alphabets of size at most 4, in the extreme edge case, our new algorithm is at least a $\sqrt{w}$ factor faster than the state of the art.
Note that since $n \le \setsize{\Sigma}^{k-1}$, for alphabets of size at most 4, this extreme edge case can only occur if $w$ is linear in $n$, where by our \Cref{thm:dir-hard-debruijn-interval} the problem is not tractable.
Specifically, in that case, the running time of both the state of the art and our novel approach become unrealistic. 

We describe how to encode a $\fptProblem$ instance $(G=(V, E), c, w)$, where $G$ is a dBG, using an auxiliary directed graph $H=(V',E')$ of size $\BigO(m \cdot \fptBase^{\frac{w}{k-1}+1})$ in which there is a path between the designated source and target nodes if and only if the instance contains a $c$-respecting Eulerian trail.
The graph $H$ consists of a sequence of $m + 1$ \emph{layers} of nodes with edges only between consecutive layers. In this model, traversing an edge of $H$ between layers $t-1$ and $t$ means traversing a corresponding edge of $G$ at time step $t$. We show the construction of graph $H=(V',E')$ and start by defining the construction of the nodes in $V'$.

The algorithms in~\cite{DBLP:journals/jcb/Ben-DorPSS02,DBLP:journals/algorithmica/BumpusM23} rely on enumerating states, where a state is described by a node $v$ of $G$ (this is the end location of the current partial trail), a time step $t$, and some information about the edges of the partial trail. Both~\cite{DBLP:journals/jcb/Ben-DorPSS02} and~\cite{DBLP:journals/algorithmica/BumpusM23} achieve their \(\fpt\) runtime by observing that it is unnecessary to encode the entire partial trail; instead, it suffices to consider the last $w$ edges of it.
They are even able to show that it is unnecessary to remember the order these edges are used on the trail.
Applying \Cref{lem:edges-per-interval} shows that, for each $v \in V$ and $t \in [m]$, it suffices to consider all subsets of a fixed size.
These are at most ${2w-1 \choose w} = \BigO(4^w / \sqrt{w})$; see \Cref{app:generalizations}. We will use additional structural insights into dBGs to severely reduce the number of states that need to be considered in many relevant cases.

Henceforth, we set \(\fptBase \coloneqq \min(\setsize{\Sigma}^{k-1}, 2w-1)\) as in \Cref{thm:fpt-dbg}.
For ease of notation, we also set $\ell \coloneqq \lceil \min(w,t) / (k-1) + 1 \rceil$ when $t$ is clear from the context. Intuitively, $\ell$ is the number of nodes required to uniquely represent a path of length $\min(w,t)$ in $G$.
By \(\id{v}\) we denote the string of length $k-1$ associated with the node \(v\in V\);
note that, we may abuse this notation and denote by \(\id{W}\) the string of length $k-1+|W|$ for a trail $W=e_1\ldots e_{|W|}$ in $G$. 

\subparagraph{The Nodes of $H$.} We add a node to $V'$ for every tuple $(t, \str)$ such that:
\begin{itemize}
    \item  $t \in [0,m]$ is a time step; and   
    \item $\str \in \Sigma^{\min(w,t)+k-1}$ is a string of length ${\min(w,t)+k-1}$ over $\Sigma$, such that, for all $i \in [\ell]$, the node of $G$ corresponding to the substring $\str[(i-1)(k-1)+1\dd i(k-1)]$ of $\str$ has an incoming edge labeled $\alpha[i(k-1)]$, available at time step $t - w + (i-1)(k-1)$.
\end{itemize}
In order to avoid edge casing, we say that the incoming-edge requirement is fulfilled vacuously for time step 0: for layers $t \le w$, we do not require the first encoded node to have an incoming edge. We also add to $V'$ auxiliary source and target nodes, denoted by $s$ and $z$, respectively.

\begin{remark}
\label{rem:fpt-node-sequence}
While this construction appears overtechnical at first glance, by the definition of a dBG $G=(V,E)$, the string \(\str\) is equivalent to a sequence of nodes $v_\ell, \dots, v_1 \in V$, such that $v_i$ has an incoming edge labeled $\id{v_i}[k-1]$ available at time step \(t - (i-1) \cdot (k-1)\).
\end{remark}

Intuitively, this construction relies on the insight that any path (of length $r$) in a dBG is fully described by the string (of length $r+k-1$) it generates and that this string in turn is fully described by examining every \((k-1)\)-th node on the path.

\subparagraph{The Edges of $H$.} For a node $(t,\str)$ in $V'$, we now potentially create an outgoing edge for each letter $x \in \Sigma$ to $(t+1, \str[2\dd |\str|] \cdot x)$, essentially deleting the leftmost letter of  $\str$ and then appending $x$. However, we add the edge $(t, \str)\overset{x}{\to}(t+1, \str[2\dd |\str|] \cdot x)$ to $E'$ only if:
\begin{itemize}
    \item both the head and tail node exist in $V'$; and
    \item $\str[|\alpha|-k+2\dd |\str|] \cdot x$ (the last $k-1$ letters of $\str$ appended with $x$) does not occur in $\str$.
\end{itemize}
The first condition ensures that the head and tail nodes fulfill the condition that the encoded nodes have incident edges at the required time steps.
The second condition ensures that the edge of $G$ with label $x$ that starts at the node $v$ with $\id{v}=\str[|\str|-k+2\dd |\str|]$ is not traversed more than once in the (Eulerian) trail we aim to construct.

We also add an edge from the source node $s$ to all nodes in layer $0$ (i.e., all nodes where $t=0$), and from all nodes in layer $m$ (i.e., all nodes where $t=m$), to the target node $z$.

See \Cref{fig:fpt-construction-str} for an illustration of this construction.

\begin{figure}[htbp]
    \centering
    \includegraphics[width=0.65\textwidth, page=3]{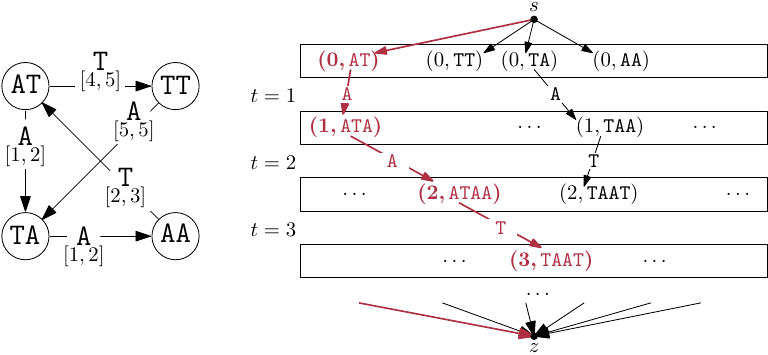}
    \caption{The input de Bruijn graph $G$ (on the left) and the construction of graph \(H\) (on the right) for $k=3$, $|\Sigma|=2$, $w=2$, and thus $w+k-1=4$.  For instance, for $t=2$, $\ell=2$, and $i=2$, we add node $(t, \alpha)=(2,\texttt{ATAA})$ to layer $t=2$ in $H$, because the node of $G$ corresponding to $\str[(i-1)(k-1)+1\dd i(k-1)]=\texttt{AA}$ has an incoming edge labeled $\str[i(k-1)]=\texttt{A}$ available at time step $t-w+(i-1)(k-1)=2$. A $c$-respecting Eulerian trail in $G$ is indicated in red in $H$.}
    \label{fig:fpt-construction-str}
\end{figure}

\begin{lemma}
    \label{lem:h-size-str}
Graph $H$ has $\BigO(m\cdot \fptBase^{\frac{w}{k-1}+1})$ nodes and $\BigO(m \setsize{\Sigma} \cdot \fptBase^{\frac{w}{k-1}+1})$ edges, where \(\fptBase \coloneqq \min(\setsize{\Sigma}^{k-1}, 2w-1)\). Graph $H$ can be constructed in time linear in its maximum size.
\end{lemma}

\begin{proof}
    Firstly, for any layer $t>0$ in $H$, there can be at most \(\setsize{\Sigma}^{\iWidth + k -1} = \left( \setsize{\Sigma}^{k-1} \right)^{\frac{\iWidth}{k-1}+1}\) choices for \(\str\).
    Secondly, by \Cref{lem:edges-per-interval}, there are at most $2w -1 $ edges available at any time step $t>0$, thus there can be at most that many nodes with at least one incoming edge available at $t$.
    Thus, by \Cref{rem:fpt-node-sequence}, there can be at most \((2w-1)^\ell\) such strings $\str$. 

    By construction, there are precisely $n$ nodes in layer $0$. Counting $s$ and $z$ as well and by the choice of \(\fptBase\), there are at most \(2+n+m \cdot \fptBase^{\frac{w}{k-1}+1}\) nodes in $H$.
    For the edges, observe that every node (except $s$) has out-degree at most $\setsize{\Sigma}$.
    Node $s$ has out-degree $n$ and node $z$ has in-degree the size of the last layer, that is at most $\fptBase^{\frac{w}{k-1}+1}$.
    This yields the bound on the number of edges. We next discuss how $H$ can be efficiently constructed.

To construct $H$ efficiently, we perform a DFS over its \emph{implicit} representation.
For each node $(t, \str)$, we store the last letter of $\str$, and associate each incoming edge with the letter it removed from the string in the previous layer to obtain \(\str\). The only part that requires a little care is when our DFS is at some node $(t,\str)$ and has to decide whether to add an edge with letter $x \in \Sigma$. In such a case, we have to check the condition that the string $\str[|\str|-k+2\dd |\str|] \cdot x$ does not occur in $\str$.
Instead of storing the entire string $\str$ and performing the pattern matching explicitly, we use a single bit-string of length $m$ during the DFS that maintains the set of $w$ previously used edges (which are equivalent to the strings of length $k$ that can appear in this form).
This set can be maintained by correctly flipping the bits corresponding to the edges at the top and at the $w$-th position in the DFS stack every time the DFS moves forward or backward.
In the word RAM model, each edge id is from $[m]$ and so it fits into one machine word. We therefore decide on the existence of each possible edge of $H$ in $\BigO(1)$ time. To conclude, constructing $H$ takes linear time in its maximum size.
\end{proof}

Formalizing the intuition behind the aim of the construction yields the following central lemma. See also \Cref{fig:fpt-construction-str} that illustrates the properties stated in the lemma.
\begin{lemma}
    \label{lem:reachable-str}
    For any $t \in [m]$ and $\str \in \Sigma^{\min(w,t)+k-1}$, the graph $H$ contains a path from $s$ to the node labeled $(t, \str)$ if and only if there exists a trail $W=e_1\ldots e_{t}$ in $G$ such that:
    \begin{enumerate}
        \item every edge from $G$ occurs at most once in $W$; \label{lem:reachable-str:euler}
        \item for every $t' \in [t]$, the edge $e_{t'}$ is available at time step $t'$; \label{lem:reachable-str:available}
        \item the last $w$ edges of $W$ correspond to $\str$; namely, they start at the node labeled $\str[1\dd k-1]$ and then follow the edges corresponding to $\str[k\dd |\str|]$. \label{lem:reachable-str:match}
    \end{enumerate}
\end{lemma}

\begin{proof}
    Let $v_\ell, \dots, v_1$ be the nodes encoded by \(\str\) under \Cref{rem:fpt-node-sequence}.
    
    We will use induction over the layers of $H$, that is, over $t$ in the statement of the lemma. For the base case, consider the nodes $(0, \str)$ that have a direct edge from $s$.
    There $\setsize{\str} = k-1$, that is, they encode a single node and no edge, thus \Cref{lem:reachable-str:match} is fulfilled.
    Given that we are looking at time step 0, the trail $W$ is empty, so it trivially exists and vacuously fulfills \Cref{lem:reachable-str:euler,lem:reachable-str:available}. 

    For the induction step, we start by proving that the existence of a path from \(s\) to \((t,\str')\) implies the three listed conditions.
    Since $H$ is arranged in layers, $(t,\str')$ must have an incoming edge from some $(t-1,\str)$ such that $(t-1,\str)$ is reachable from $s$.
    By the induction hypothesis (and \Cref{lem:reachable-str:match}), there exists a valid trail $W$ that ends in the node $u_1$ corresponding to the last $k-1$ letters in $\str$ after $t-1$ time steps.
    We will show that the trail $W'$ obtained by appending the outgoing edge from $u_1$ labeled $\str'[|\str'|]$
    (the last letter of $\str'$) to $W$ fulfills the given conditions.
    Let $(u_1, v_1)$ be that edge.
    Note that \(\id{v_1} = \str[|\str|-k+2\dd |\str|]\) (i.e., $v_1$ is precisely the node encoded by the last $k-1$ letters of $\str'$).

    By the construction of the nodes of $H$, the edge $u_1 \overset{x}{\to} v_1$ is available at $t$.
    For all other edges on $W'$, the analogous is true by the induction hypothesis.
    Thus, \Cref{lem:reachable-str:available} holds for $W'$.
    
    For \Cref{lem:reachable-str:euler}, we show that $(u_1,v_1)$ is not in $W$.
    All other edges are unique by the induction hypothesis.
    By the definition of $w$, it can thus only be available after $t-w+1$.
    Thus, if it were to appear on $W$, it must appear as one of the $w-1$ last edges.
    Then the string $\id{u_1} \cdot x$ is a substring of $\str'$, which contradicts our construction of the edges in $H$, as there would then be no edge from $(t-1, \str)$ to $(t,\str')$.
    Thus, $(u_1, v_1)$ does not appear in $W$ and \Cref{lem:reachable-str:euler} holds.

    \Cref{lem:reachable-str:match} follows since $\id{W}[|\id W|-w-k+2\dd |\id W|] \cdot x = \str \cdot x = \id{W}[|\id W|-w-k+2] \cdot \id{W'}[|\id{W'}|-w-k+2\dd |\id{W'}|] = \id{W}[|\id W|-w-k+2] \cdot \str'$ and thus \(\str' = \id{W'}[|\id{W'}|-w-k+2 \dd |\id{W'}|]\), which is the required property.

    The other direction of the implication holds by verifying that such a trail fits the construction of $H$.
\end{proof}

From \Cref{lem:reachable-str}, we can deduce the following crucial property, which both proves the correctness of our algorithm (i.e., \Cref{thm:fpt-dbg}) and underpins our extensions.

\begin{lemma}
    \label{lem:s-z-path-str}
    Graph $H$ contains an $s$-to-$z$ path if and only if there exists an Eulerian trail $W=e_1 \ldots e_m$ in $G$, such that for every $i \in [m]$, edge $e_i$ is available at time step $i$.
\end{lemma}
\begin{proof}
    By construction of $H$, there are edges from all nodes in layer $m$ to node $z$. \Cref{lem:reachable-str} tells us that there is a path from $s$ to $(m,\str)$ if and only if there exists a trail $W$ of $m$ edges in $G$ such that:
    \begin{enumerate}
        \item every edge occurs at most once in $W$,
        \item for every $t \in [m]$, the $t$-th edge is available at time step $t$.
    \end{enumerate}
It follows immediately that $W$ is an Eulerian trail, as it has length $m$. Since each $(m,\str)$ node has an edge going into $z$, it follows that there is an $s$-to-$z$ path in $H$ if and only if there exists a valid Eulerian trail in $G$.
\end{proof}

\Cref{lem:h-size-str} and \Cref{lem:s-z-path-str}
imply \Cref{thm:fpt-dbg}. We stress that, by \Cref{lem:h-size-str}, 
if we waive the assumption that $|\Sigma|=\BigO(1)$, we get an algorithm that is slower only
by a factor of $|\Sigma|$.

\subparagraph{Generalizations.} We now discuss how to exploit this construction to solve more general problems.
These extensions work analogously on the algorithm underlying \Cref{thm:fpt} for general directed graphs.
As a first extension, we can use these algorithms to report a witness: a $c$-respecting Eulerian trail in $G$.
The following result makes this precise.

\begin{corollary}
\label{thm:finding-alg}
For any instance \((G, c)\), where \(c\) is an interval function with width \(w\), there is an $\BigO(m \cdot \iWidth^{1.5} 4^{\iWidth})$-time algorithm finding a $c$-respecting Eulerian trail in \(G\) if it exists and correctly deciding non-existence otherwise.
On a de Bruijn graph of order \(k\) over alphabet \(\Sigma\), \(|\Sigma|=\BigO(1)\), we can solve this problem in $\BigO(m \cdot \lambda^{\frac{\iWidth}{k-1}+1})$ time, where \(\fptBase \coloneqq \min(\setsize{\Sigma}^{k-1}, 2w-1)\).
\end{corollary}

\begin{proof}
For general graphs, the result follows from \Cref{lem:s-z-path}.
The Eulerian trail $W$ constructed in the proof of the lemma is easy to compute by taking the $s$-to-$z$ path from $H$ and replacing every edge $(t-1, u, S) \to (t, v, S')$ with the edge $u \to v \in G$.

For dBGs, the result follows from \Cref{lem:s-z-path-str}.
The Eulerian trail $W$ constructed in the proof of the lemma is easy to compute by taking the $s$-to-$z$ path from $H$ and replacing every edge $(t, u, \str) \to (t+1, v, \str')$ with the edge $u \to v \in G$.
\end{proof}

For small $w \log w/(k-1)$ values, the algorithm on dBGs is highly efficient.
At first glance, that is unintuitive, as even if for each pair of consecutive positions there are only two available $k$-mers (edges), there could already be $2^{m/2}$ possible orderings.  
Yet, our result shows that the additional structure of the graph allows us to efficiently select the correct ordering. The following remark captures the significant improvement over previous works which required $w$ to be tractably small to efficiently solve the string reconstruction problem.

\begin{remark}
    Let \(|\Sigma|=\BigO(1)\).
    For any $\frac{w \cdot \log w}{k-1}= \BigO(1)$, the $\fptProblem$ problem on dBGs can be solved in \(\BigO(mw)\) time. For any $w= \BigO(1)$, $\fptProblem$ on dBGs can be solved in \(\BigO(m)\) time. 
\end{remark}

It is important to note that in many practical applications, especially in bioinformatics and data privacy (see \Cref{fig:dBG}), the dBGs considered are multigraphs (have parallel edges).
For our $\fpt$ technique, the analysis above relies on the fact that an edge is uniquely determined by its endpoints, but with some extra care it can be adjusted to extend to multigraphs. We give a full discussion in \Cref{sec:multigraphs}, but we restate the following main observation here.
\begin{restatable}{observation}{obsMultigraphsWork}
If there is a $c$-respecting Eulerian trail in a de Bruijn multigraph $G=(V,E)$, with $m=|E|$, then there is a $c$-respecting Eulerian trail $W=e_1 \dots e_{m}$ such that for any time step $t$ where the edge $e_t$ has a parallel edge $e'$ that is also available at $t$ and appears after $e_t$ on $W$, we have that $\max c(e_t) \le \max c(e')$. 
\end{restatable}

Moreover, as mentioned above, we can generalize our technique to interval cost functions and retain the same running time, for both general directed graphs and for dBGs.

\findingAlgGen

\begin{proof}
    To see this is a corollary of \Cref{thm:finding-alg}, notice the important change when compared to interval functions: each time step is associated with a cost.

    For general graphs, observe that every Eulerian trail consisting only of available edges in $G$ uniquely corresponds to an $s$-to-$z$ path $W$ in $H$ by \Cref{lem:s-z-path}.
    If we now modify the construction of $H$ such that an edge $(t-1, v, S)$ to $(t, u, S')$ has cost $c(vu, t)$, each such path $W=\{(v_i, i, S_i)\}_{i \in [0,m]}$ has cost precisely $\sum_{i\in [m]}{c(v_{i-1}v_i, i)}$, which by definition is the cost of the corresponding Eulerian trail.
    Thus a min-cost Eulerian trail in $G$ corresponds now to the min-cost $s$-to-$z$ path in $H$.
    Note that since $H$ is a DAG, we can find a shortest $s$-to-$z$ path in graph-linear time, thus we do not asymptotically increase the running time when compared to the unweighted version.
    Having addressed this issue, the desired result follows.

    For dBGs, every Eulerian trail consisting of only available edges in $G$ uniquely corresponds to an $s$-to-$z$ path $W$ in $H$ by \Cref{lem:s-z-path-str}.
    If we modify the construction of $H$ such that an edge $(t-1, \str)$ to $(t, \str')$ has the cost associated with the last encoded edge $e_i$ in $G$, each such path $W=\{(i, \str_i)\}_{i \in [0,m]}$ has cost $\sum_{i\in [m]}{c(e_i, i)}$, where $\id{e_i} = \str_i[|\str_i|-k+1\dd |\str|]$, which by definition is the cost of the corresponding Eulerian trail.
    Thus a min-cost Eulerian trail in $G$ corresponds to the min-cost $s$-to-$z$ path in $H$.
    Finally, since $H$ is a DAG, we can find a shortest $s$-to-$z$ path in linear time, thus asymptotically we retain the running time of the unweighted setting.
\end{proof}

Note that in this setting, we can allow an edge to become temporarily unavailable by setting its cost to at least $\cmax+1$.
In such a case, the width of the edge's interval remains the difference between the last and the first time step at which the edge is available.

\Cref{thm:finding-alg-gen} directly implies that $\dirProblemCost$ parametrized by $\iWidth$ is in $\fpt$.
On dBGs, $\dirProblemCost$ is also in \(\fpt\) when parametrized by $\iWidth \log \iWidth /(k-1)$.

\section{Counting Eulerian Trails with Applications in Data Privacy}
\label{sec:counting}
Our results are of crucial relevance for the data privacy applications presented in~\cite{DBLP:conf/alenex/0001CFLP20} and~\cite{DBLP:journals/jea/BernardiniCFLP21}.
On the one hand, our hardness results (\Cref{thm:dir-hard-debruijn-interval} and \Cref{cor:inapprox}) show that \emph{in general} it is indeed hard to reconstruct a private (binary) string from a given dBG, thus providing theoretical justification for the privacy model introduced in~\cite{DBLP:conf/alenex/0001CFLP20}.
On the other hand, our algorithmic results (\Cref{thm:fpt-dbg} and \Cref{thm:finding-alg-gen}) show that, with sufficiently specific domain knowledge, one can reconstruct such a (best) string in linear time.
To address this concern, the algorithms of~\cite{DBLP:conf/alenex/0001CFLP20,DBLP:journals/jea/BernardiniCFLP21} rely on a counting routine to efficiently ensure that the released string has at least $\kappa$ possible reconstructions.
To verify that the $\kappa$-anonymity property continues to hold even when reconstruction is aided by domain knowledge, we can directly employ the following result as the counting routine in the main algorithm of~\cite{DBLP:conf/alenex/0001CFLP20,DBLP:journals/jea/BernardiniCFLP21}.

\thmAlgCounting
\begin{proof}
    Observe that by \Cref{lem:s-z-path-str} and \Cref{thm:finding-alg-gen}, there is a one-to-one correspondence between min-cost $s$-to-$z$ paths in $H$ and min-cost $c$-respecting Eulerian trails in $G$.
    Thus, we can apply the folklore linear-time shortest-path counting DAG dynamic program~\cite{DBLP:journals/mst/MihalakSW16}.
\end{proof}

\begin{remark}
The algorithm underlying \Cref{thm:alg-counting} can be trivially extended to enumerate all such Eulerian trails in additional time that is linear in the size of the output.
\end{remark}

\paragraph{Acknowledgments.}
SPP is supported in part by the PANGAIA and ALPACA projects that have received funding from the European Union’s Horizon 2020 research and innovation programme under the Marie Skłodowska-Curie grant agreements No 872539 and 956229, respectively.
LS is supported in part by the Netherlands Organisation for Scientific Research (NWO) through project OCENW.GROOT.2019.015 ``Optimization for and with Machine Learning (OPTIMAL)'' and Gravitation-project NETWORKS-024.002.003.
HV is supported by a Constance van Eeden Fellowship.

\bibliography{papers}
\newpage
\appendix
\section{Generalizations of the Ben{-}Dor et al.~and Bumpus and Meeks Algorithm}\label{app:generalizations}
First, we discuss our notation for the algorithm.
This also shows how to shave off a factor of $w$ in exchange for computing an Eulerian trail instead of an Eulerian cycle compared the the Bumpus and Meeks version.
For completeness, we give a correctness proof, but the argument is entirely analogous to the one presented in~\cite{DBLP:journals/jcb/Ben-DorPSS02} and~\cite{DBLP:journals/algorithmica/BumpusM23}.

\begin{lemma}\label{lem:available-edge-set-size}
    Let $t\in [m]$ be a time step.
    The total number of sets $S\subseteq E(t)$, for which there exists some Eulerian trail $P=e_1\ldots e_m$, such that $\{  e_i\in P \mid i\in[0,t], e_i \in E(t) \} = S$, is in $\BigO(4^w/\sqrt{w})$.
    Furthermore, all such sets $S$ have size $t-a$, where $a$ is the number of edges which will not be available at or after $t$.
\end{lemma}
\begin{proof}
    Call an edge $e$ expired at $t$ if $\max c(e) \le t$.
    Let $\mathcal{S}(t)$ be the set of sets $S$ such that there exists some Eulerian trail $P=e_1\ldots e_m$, with $\{  e_i\in P \mid i\in[0,t], e_i \in E(t) \} = S$.
    We will show that any such $S\in \mathcal{S}(t)$ has size $t-a$, where $a$ is the number of edges expired before $t$.
    Let $S$ be an arbitrary set in $\mathcal{S}(t)$ and let $P$ be an Eulerian trail $P=e_1\ldots e_m$ such that $S=\{  e_i\in P \mid i\in[0,t], e_i \in E(t) \}$ .
    Then, at time $t$, $P$ must have traversed all edges that expire before time $t$.
    Since, at time step $t$, $t$ steps have passed, $\setsize{S} = \setsize{\{i\in [0,t] \mid e_i\in E(t)\}} = t-a$.
    $S$ was taken arbitrarily, so all such sets have size $t-a$.
    This means that $\mathcal{S}(t)\subseteq \{T\subseteq E(t) \mid |T|=t-a\}$ and thus $|\mathcal{S}(t)|\leq {|E(t)|\choose t-a}$.
    We also know by \Cref{lem:edges-per-interval}, that $|E(t)|\leq 2w-1$, so $|\mathcal{S}(t)|\leq {2w-1\choose(t-a)}$.
    Therefore, $|\mathcal{S}(t)|\leq {2w-1 \choose w} \sim \frac{4^w}{\sqrt{w}} \in \BigO( \frac{4^w}{\sqrt{w\pi}}) $.

\end{proof}

We now define the construction of the nodes in $V'$.
We create a node in $V'$ for every triple $(t, v, S)$ such that
\begin{itemize}
    \item  $t \in [0,m]$ is a time step,
    \item $v \in V$ is a node from $G$ which has at least one incoming edge available at $t$,
    \item $S \subseteq E(t)$, if $|S|=t-a(t)$, where $a(t)$ is the number of edges that have expired at time $t$
    (for simplicity, we say $E(0)=\emptyset$), and
    \item if $t>0$, there is some edge in $S$ that ends in $v$
\end{itemize}
We also create auxiliary source and target nodes $s$ and $z$, respectively.

Construct the edges in $E'$ as follows. Every edge either connects two nodes with consecutive time steps, or it connects $s$ to a node with time step 0 or a node with time step $m$ to $z$.
For the source and target nodes, we add the edges $(s,(0,v,\emptyset))$ and $((m,v,S),z)$ for every node $v \in V$ and $S \subseteq E(m)$.
We define the edges between consecutive time steps as follows; refer to \Cref{fig:fpt-construction} for an illustration of the construction.

\begin{figure}[ht]
    \centering
    \includegraphics[height=5cm, page=2]{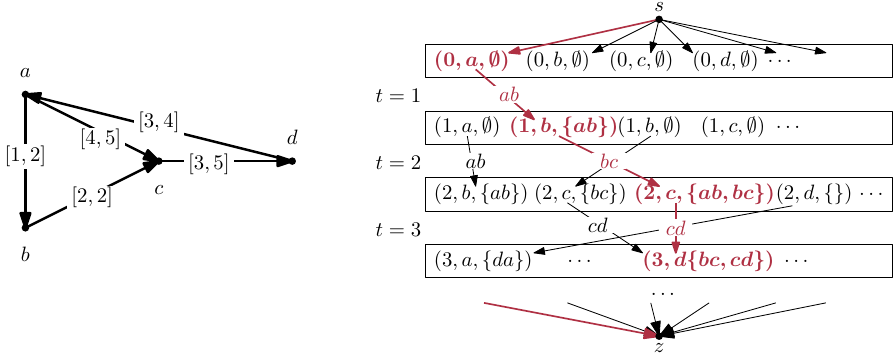}
    \caption{The construction of graph \(H\) (on the right) for the graph $G$ (on the left). The Eulerian trail in $G$ is indicated in red in $H$. Some of the nodes and edges in $H$ are omitted for legibility.}
    \label{fig:fpt-construction}
\end{figure}

\begin{definition}
    \label{def:h-edge}
    Let $(t-1, u, S)$ and $(t, v, S')$ be two nodes from $V'$. We add the edge $((t-1,u,S),(t,v,S'))$ to $E'$ if and only if
    \begin{enumerate}
        \item $(u,v)$ is an edge in $G$ and available at time step $t$, \label{def:h-edge:available}
        \item $(u,v) \notin S$, and \label{def:h-edge:not-picked}
        \item $S'=(S\cap E(t))\cup\{(u,v)\}$. \label{def:h-edge:update-s}
    \end{enumerate}
\end{definition}

Note that by constructing $H$ using an implicit DFS, we can exclude nodes that do not have an incoming edge.
In the next lemma, we prove an upper bound on both the number of nodes and edges of $H$ if it is created this way, as as functions of $n,m$ and $w$.

\begin{lemma}
    \label{lem:h-size}
    Graph $H$ contains at most $2+n+m\cdot \BigO(\sqrt{w} \cdot 4^w)$ nodes and $2n + m\cdot \BigO(w^{1.5} \cdot4^{w})$ edges, where $n$ and $m$ are the number of nodes resp. edges in $G$.
\end{lemma}
\begin{proof}
    By \Cref{lem:edges-per-interval}, we have that $|E(t)| \leq 2w - 1$ for every time step $t > 0$.
    Every node other than $s$ and $z$ is labeled with a time step $t\in [0,m]$, a node $v \in V$ and a set $S \subseteq E(t)$.
    Moreover, $v$ must be the endpoint of one of the edges in $E(t)$, meaning the number of choices for $v$ comes down to selecting one of the at most $2w - 1$ edges in $E(t)$.
    By \Cref{lem:available-edge-set-size}, the number of sets $S\subseteq E(t)$ we have to consider is bounded by $\BigO(4^w/\sqrt{w})$.
    Thus, for time steps greater than zero, there are at most $\BigO(\sqrt{w} \cdot 4^w)$ distinct triples.
    For nodes with time step zero, we have $E(0)=\emptyset$, and so there are just $n$ such nodes.
    Adding nodes $s$ and $z$ gives a total node count of $2 + n + m \cdot \BigO(\sqrt{w} \cdot 4^{w})$.

    For the number of edges, we first consider the outgoing edges of all labeled nodes with a time step less than $m$.
    For a given node $(t,u,S)$, all outgoing edges go to nodes labeled $(t,v,(S\cap E(t+1))\cup\{(u,v)\})$ for some $(u,v)\in E(t+1)$.
    Since by \Cref{lem:edges-per-interval}, $|E(t+1)| \leq 2w - 1$, we have at most $2w - 1$ outgoing edges for each such labeled node.
    All other edges in $H$ are incident on either $s$ or $z$, with $n$ such edges for each, giving a total edge count of $2n + m\cdot \BigO(w^{1.5} \cdot2^{w})$.
\end{proof}

Let us remark that given a node $(t,v,S)$ of $H$, its outgoing edges can easily be determined by iterating over $v$'s outgoing edges (that are not already in $S$) and checking their availability at time step $t+1$.

In the next lemma, we prove a one-to-one correspondence between length-$t$ paths in $H$ starting from the source node and length-$t$ trails in $G$:

\begin{lemma}
    \label{lem:reachable}
    For any $t \in [m], v \in V, S \subseteq E(t)$, we have that $H$ contains a path from $s$ to the node labeled $(t, v, S)$ if and only if there exists a trail $W=e_1\ldots e_{t}$ in $G$ such that
    \begin{enumerate}
        \item every edge from $G$ occurs at most once in $W$, \label{lem:reachable:euler}
        \item for every $t$, the edge $e_t$ is available at time step $t$, \label{lem:reachable:available}
        \item $W$ ends in $v$ unless it is empty (either $v$ is the endpoint of $e_{t}$ or $t=0$), and \label{lem:reachable:endpoint}
        \item $W \cap E(t) = S$ (i.e., $S$ contains all available but already used edges). \label{lem:reachable:s}
    \end{enumerate}
\end{lemma}
\begin{proof}
    We use induction over the layers of $H$, that is over $t$ in the statement of the lemma. For the base case, consider the nodes $(0, v, \emptyset)$ which have a direct edge from $s$.
    Given that we are looking at time step 0, the trail $W$ is empty so it trivially exists. Moreover, $W \cap E(0) = \emptyset$ fulfilling \Cref{lem:reachable:s}.

    For the induction step, we start by proving that the existence of a path from \(s\) to \((t,v, S')\) implies the four listed conditions.
    By this assumption, $(t,v,S')$ has an incoming edge from some $(t-1,u,S)$ such that $(t-1,u,S)$ is reachable from $s$.
    By the induction hypothesis, there exists a valid trail $W$ that ends in $u$ after $t-1$ time steps, with $W \cap E(t-1) = S$.
    We will show that the trail $W'$ obtained by appending $(u,v)$ to $W$ fulfills the given conditions.

    For \Cref{lem:reachable:euler}, we show that $(u,v)$ is not in $W$.
    All other edges are unique by the induction hypothesis.
    Note that because there is an edge from $(t-1,u,S)$ to $(t,v,S')$, it must be that $(u,v)\notin S$; otherwise, this edge would not fulfill \Cref{def:h-edge:not-picked} from \Cref{def:h-edge}.
    Similarly, $(u,v)$ must be available at time step $t$ by \Cref{def:h-edge:available} of \Cref{def:h-edge}.
    Now assume $(u,v)$ appears at some earlier time step $t'$ on $W$.
    Then, because we are considering interval functions, it must be available at $t-1$ as well since $t' \le t-1 <t$.
    Thus $(u,v) \in W \cap E(t-1) = S$, which violates \Cref{def:h-edge:not-picked} of \Cref{def:h-edge}.

    \Cref{lem:reachable:available} holds trivially due to the corresponding edge existing in $H$.
    \Cref{lem:reachable:endpoint} also holds trivially, as the final edge of $W'$ is $(u,v)$.

    For \Cref{lem:reachable:s}, combine that we already know by the induction hypothesis that $W \cap E(t-1)=S$ with the fact that, by the construction of $H$, we have $S'=(S\cap E(t))\cup\{(u,v)\}$.
    
    To prove the implication in the other direction, suppose there exists some (non-empty) trail $W'=e_1\ldots e_t$ for which the listed conditions hold. Let $t \geq 1$, $e_t = (u,v)$ and $S' = W' \cap E(t)$.
    We will show that the node $(t, v, S')$ is reachable in $H$.
    Let $W = W' \setminus \{(u,v)\}$ and let $S = W \cap E(t-1)$. $W$ and $S$ fulfill the four conditions, so by the induction hypothesis, the node $(t-1,u,S)$ is reachable in $H$.
    We will now show that the edge $((t-1,u,S),(t,v,S'))$ exists in $H$ using the conditions from \Cref{def:h-edge}.

    \Cref{def:h-edge:available,def:h-edge:not-picked} are trivially true.
    For \Cref{def:h-edge:update-s}, we have to show that $S'=(S\cap E(t))\cup\{(u,v)\}$, which holds since $S ' = W' \cap E(t ) = (W \cap E(t-1) \cap E(t)) \cup \{(u,v)\} = (S \cap E(t)) \cup \{(u,v)\}$.
    With all three conditions proven, the edge $((t-1,u,S),(t,v,S'))$ must exist, making the latter node reachable.
 \end{proof}
\begin{lemma}
    \label{lem:s-z-path}
    $H$ contains an $s$-to-$z$ path if and only if there exists an Eulerian trail $W=e_1\ldots e_m$ in $G$, such that for every $i \in [m]$, edge $e_i$ is available at time step $i$.
\end{lemma}
\begin{proof}
    By construction of $H$, there are edges from all nodes $(m, v, E(m))$ to $z$. \Cref{lem:reachable} tells us that there is a path from $s$ to $(m,v,E(m))$ if and only if there exists a trail $W$ of $m$ edges in $G$ such that:
    \begin{enumerate}
        \item every edge occurs at most once in $W$,
        \item for every $t \in [m]$, the $t$-th edge is available at time step $t$,
        \item $W$ ends in node $v$, and
        \item $E(m) \subseteq W$.
    \end{enumerate}
It follows immediately that $W$ is an Eulerian trail, as it has length $m$. Since each $(m,v,E(m))$ node has an edge going into $z$, it follows that there is an $s$-to-$z$ path in $H$ if and only if there exists a valid Eulerian trail in $G$.
\end{proof}

\begin{proof}[Proof of \Cref{thm:fpt}]
    We solve the problem by constructing the graph $H$ and checking if there exists an $s$-to-$z$ path.
    By \Cref{lem:h-size}, $H$ takes $\BigO(m \cdot w^{1.5} 4^w)$ space, and it can be constructed in the same time by first creating all the nodes and then creating the outgoing edges for each node by iterating over its available edges.
    Identifying an $s$-to-$z$ path is possible using, for example, depth-first search, which takes time linear in the size of $H$.
    By \Cref{lem:s-z-path}, there exists a valid Euler trail in $G$ if and only if we find an $s$-to-$z$ path in $H$; after deciding if it exists, we can retrieve the trail by observing the labels of the nodes on the $s$-to-$z$ path.
\end{proof}

\section{Extensions to Multigraphs}\label{sec:multigraphs}
So far we have only considered simple graphs, but especially in the setting of string reconstruction, dBGs are often applied as multigraphs; see \Cref{fig:dBG}.
Notably, most  of our algorithmic results translate to this important setting.
First of all, we observe that the state-of-the-art algorithm by Ben{-}Dor et al.~\cite{DBLP:journals/jcb/Ben-DorPSS02} and Bumpus and Meeks~\cite{DBLP:journals/algorithmica/BumpusM23} does not use the fact that edges are uniquely defined by their endpoints. Thus this algorithm, as well as the optimization and counting extensions we  introduce, translate immediately to multigraphs.

For our novel algorithm exploiting the dBG structure, we must be more careful, as this approach crucially relies on identifying edges by their corresponding string from $\Sigma^k$ (or equivalently by their associated letter and start- or end-node). Luckily, the interval availabilities add sufficient structure, that we can use an exchange-style argument to overcome the challenges raised by this non-unique identification.
This relies on the following:
\obsMultigraphsWork*
Simply put, presented with a choice between two parallel edges that are both available at $t$, it suffices to greedily choose the edge whose interval ends earlier.
This holds since switching any two parallel edges violating the condition again yields a $c$-respecting Eulerian trail. 

Taking care that the bit-string which tracks used edges in our novel algorithm chooses which copy of a multiedge to mark according to this rule (breaking ties consistently), the correctness of the algorithm on multigraphs follows.
Unfortunately, this does not translate to the optimization variant as the exchange argument does not preserve the trail's cost, unless we enforce that parallel edges must have the same cost at any time step they share.

Since by the above observation, \Cref{lem:s-z-path} translates to multigraphs saying that each $s$-to-$z$ path in $H$ corresponds to an Eulerian trail whose node sequence is unique, the counting algorithm also works in that setting.
Depending on the application, this might even be the more interesting object to count.
At least in the string reconstruction setting, counting Eulerian trails that yield different strings is the natural objective~\cite{DBLP:conf/fct/ConteGLPPP21,DBLP:journals/jea/BernardiniCFLP21,DBLP:conf/alenex/0001CFLP20}.
For example, in the data privacy application discussed in \Cref{sec:counting}, different Eulerian trails which yield the same string reconstruction do not contribute to the privacy goal.

\end{document}